\documentclass{article}

\usepackage{times}
\usepackage{amssymb,amsmath,amsthm}
\usepackage{mathrsfs}
\usepackage{epsfig}
\usepackage{color}

\newcommand\ie{{\em i.e.}~}
\newcommand\eg{{\em e.g.}~}

\newcommand\aew{{\em a.e.}~}

\def\B{\mathscr B}

\def\C{\mathbb C}
\def\D{\mathscr D}
\def\DD{\mathsf D}

\def\F{\mathscr F}
\def\G{\mathcal G}
\def\H{\mathcal H}
\def\HH{\mathsf H}

\def\K{\mathcal K}
\def\KK{\mathscr K}

\def\M{\mathsf M}
\def\N{\mathbb N}

\def\R{\mathbb R}
\def\S{\mathscr S}

\def\TT{\mathsf T}
\def\U{\mathscr U}

\def\Z{\mathbb Z}
\def\v{\varphi}

\def\12{{\textstyle\frac12}}
\def\<{\left\langle}
\def\>{\right\rangle}
\def\({\left(}
\def\){\right)}
\def\[{\left[}
\def\]{\right]}
\def\dom{\mathcal D}
\def\lone{\mathsf{L}^{\:\!\!1}}

\def\ltwo{\mathsf{L}^{\:\!\!2}}
\def\linf{\mathsf{L}^{\:\!\!\infty}}

\def\e{\mathop{\mathrm{e}}\nolimits}
\def\d{\mathrm{d}}

\def\Aut{\mathop{\mathrm{Aut}}\nolimits}
\def\Hom{\mathrm{Hom}}
\def\im{\mathop{\mathsf{Im}}\nolimits}
\def\supp{\mathop{\mathrm{supp}}\nolimits}
\def\sgn{\mathop{\mathrm{sgn}}\nolimits}

\def\slim{\hbox{s-}\lim}

\def\im{\mathop{\mathsf{Im}}\nolimits}
\def\re{\mathop{\mathsf{Re}}\nolimits}
\def\Hrond{\mathscr H}
\def\ad{\mathop{\mathrm{ad}}\nolimits}

\newtheorem{Theorem}{Theorem}[section]
\newtheorem{Remark}[Theorem]{Remark}
\newtheorem{Lemma}[Theorem]{Lemma}
\newtheorem{Assumption}[Theorem]{Assumption}
\newtheorem{Corollary}[Theorem]{Corollary}
\newtheorem{Proposition}[Theorem]{Proposition}\newtheorem{Definition}[Theorem]{Definition}

\begin{document}


\title{On a new formula relating\\localisation operators to time operators}

\author{S. Richard$^1\footnote{On leave from Universit\'e de Lyon;
Universit\'e Lyon 1; CNRS, UMR5208, Institut Camille Jordan, 43 blvd du 11 novembre 1918, F-69622 Villeurbanne-Cedex, France.}\ $ and R.
Tiedra de Aldecoa$^2$}
\date{\small}
\maketitle \vspace{-1cm}

\begin{quote}
\emph{
\begin{itemize}
\item[$^1$] Department of Pure Mathematics and Mathematical Statistics,
Centre for Mathematical Sciences, University of Cambridge,
Cambridge, CB3 0WB, United Kingdom
\item[$^2$] Facultad de Matem\'aticas, Pontificia Universidad Cat\'olica de Chile,\\
Av. Vicu\~na Mackenna 4860, Santiago, Chile
\item[] \emph{E-mails:} sr510@cam.ac.uk,
rtiedra@mat.puc.cl
\end{itemize}
  }
\end{quote}


\begin{abstract}
We consider in a Hilbert space a self-adjoint operator $H$ and a family
$\Phi\equiv(\Phi_1,\ldots,\Phi_d)$ of mutually commuting self-adjoint operators.
Under some regularity properties of $H$ with respect to $\Phi$, we propose two new
formulae for a time operator for $H$ and prove their equality. One of the expressions
is based on the time evolution of an abstract localisation operator defined in terms
of $\Phi$ while the other one corresponds to a stationary formula. Under the same assumptions, we also conduct the spectral analysis of $H$ by using the method of the conjugate operator.

Among other examples, our theory applies to Friedrichs Hamiltonians, Stark
Hamiltonians, some Jacobi operators, the Dirac operator, convolution operators on
locally compact groups, pseudodifferential operators, adjacency operators on graphs
and direct integral operators.
\end{abstract}



\section{Introduction and main results}\label{Intro}
\setcounter{equation}{0}

Let $H$ be a self-adjoint operator in a Hilbert space $\H$ and let $T$ be a linear
operator in $\H$. Generally speaking, the operator $T$ is called a time operator for
$H$ if it satisfies the canonical commutation relation
\begin{equation}\label{veryformal}
[T,H]=i,
\end{equation}
or, alternatively, the relation
\begin{equation}\label{formal}
T\e^{-itH}=\e^{-itH}(T+t).
\end{equation}
Obviously, these two equations are very formal and not equivalent. So many authors
have proposed various sets of conditions in order to give a precise meaning to them.
For instance, one has introduced the concept of infinitesimal Weyl relation in the
weak or in the strong sense \cite{JM80}, the $T$-weak Weyl relation \cite{Miy01} or
various generalised versions of the Weyl relation (see \eg \cite{Ara05,HKM09}). However,
in most of these publications the pair $\{H,T\}$ is a priori given and the authors
are mainly interested in the properties of $H$ and $T$ that can be deduced from a
relation like \eqref{formal}. In particular, the self-adjointness of $T$, the
spectral nature of $H$ and $T$, the connection with the survival probability, the
form of $T$ in the spectral representation of $H$, the relation with the theory
of irreversibility and many other properties have been extensively discussed in the
literature (see \cite[Sec.~8]{ML00}, \cite[Sec.~3]{MSE08},
\cite{APSS00,Gal02,Gom08,GYS81,WX07} and references therein).

Our approach is radically different. Starting from a self-adjoint operator $H$, one
wonders if there exists a linear operator $T$ such that \eqref{veryformal} holds in
a suitable sense. And can we find a universal procedure to construct such an
operator\,? This paper is a first attempt to answer these questions.

Our interest in these questions has been recently aroused by a formula put into
evidence in \cite{Tie09_3}. Along the proof of the existence of time delay for
hypoelliptic pseudodifferential operators $H:=h(P)$ in $\ltwo(\R^d)$, the author
derives an integral formula linking the time evolution of localisation operators to
the derivative with respect to the spectral parameter of $H$. The formula reads as
follows: if $Q$ stands for the family of position operators in $\ltwo(\R^d)$ and $f:\R^d\to\C$ is some appropriate function with $f=1$ in a neighbourhood of $0$,
then one has on suitable elements $\v\in\ltwo(\R^d)$
\begin{equation}\label{eq_sympa}
\lim_{r\to\infty}\12\int_0^\infty\d t\,\big
\langle\v,\big[\e^{-itH}f(Q/r)\e^{itH}-\e^{itH}f(Q/r)\e^{-itH}\big]\v\big\rangle
=\textstyle\big\langle\v,i\frac\d{\d H}\;\!\v\big\rangle,
\end{equation}
where $\frac\d{\d H}$ stands for the operator acting as $\frac\d{\d\lambda}$ in the
spectral representation of $H$. So, this formula furnishes a standardized procedure
to obtain a time operator $T$ only constructed in terms of $H$, the position
operators $Q$ and the function $f$.

A review of the methods used in \cite{Tie09_3} suggested to us that Equation \eqref{eq_sympa} could be extended to the case of an abstract pair of operator $H$ and position
operators $\Phi$ acting in a Hilbert space $\H$, as soon as $H$ and $\Phi$
satisfy two appropriate commutation relations. Namely, suppose that you are
given a self-adjoint operator $H$ and a family $\Phi\equiv(\Phi_1,\ldots,\Phi_d)$
of mutually commuting self-adjoint operators in $\H$. Then, roughly speaking, the
first condition requires that for some $\omega\in\C\setminus\R$ the map
$$
\R^d\ni x\mapsto\e^{-ix\cdot\Phi}(H-\omega)^{-1}\e^{ix\cdot\Phi}\in\B(\H)
$$
is $3$-times strongly differentiable (see Assumption \ref{chirimoya} for a precise statement). The second condition, Assumption \ref{commute}, requires that for each
$x\in\R^d$, the operators $\e^{-ix\cdot \Phi}H\e^{ix\cdot \Phi}$ mutually commute.
Given this, our main result reads as follows (see Theorem \ref{for_Schwartz} for a
precise statement):

\begin{Theorem}\label{version_intro}
Let $H$ and $\Phi$ be as above. Let $f$ be a Schwartz function on $\R^d$ such that
$f=1$ on a neighbourhood of $0$ and $f(x)=f(-x)$ for each $x\in\R^d$. Then, for
each $\v$ in some suitable subset of $\H$ one has
\begin{equation}\label{Eq_prin}
\lim_{r\to\infty}\12\int_0^\infty\d t\,\big\langle\v,
\big[\e^{-itH}f(\Phi/r)\e^{itH}-\e^{itH}f(\Phi/r)\e^{-itH}\big]\v\big\rangle\\
=\langle\v,T_f\v\rangle,
\end{equation}
where the operator $T_f$ acts, in an appropriate sense, as $i\frac\d{\d\lambda}$ in
the spectral representation of $H$.
\end{Theorem}

One infers from this result that the operator $T_f$ is a time operator. Furthermore,
an explicit description of $T_f$ is also available: if $H_j'$ denotes the self-adjoint operator associated with the commutator $i[H,\Phi_j]$ and $H':=(H_1',\ldots,H_d')$,
then $T_f$ is formally given by
\begin{equation}\label{surTf}
T_f=-\12\big(\Phi\cdot R_f'(H')+R_f'(H')\cdot\Phi\big),
\end{equation}
where $R_f':\R^d\to\C^d$ is some explicit function (see Section \ref{Averaged} and
Proposition \ref{lemma_T_f}).

In summary, once a family of mutually commuting self-adjoint operators
$(\Phi_1,\dots,\Phi_d)$ satisfying Assumptions \ref{chirimoya} and \ref{commute}
has been given, then a time operator can be defined either in terms of the l.h.s.
of \eqref{Eq_prin} or in terms of \eqref{surTf}. When suitably defined, both
expressions lead to the same operator. We also mention that the equality
\eqref{Eq_prin}, with r.h.s. defined by \eqref{surTf}, provides a crucial
preliminary step for the proof of the existence of quantum time delay and
Eisenbud-Wigner Formula for abstract scattering pairs $\{H,H+V\}$. In addition,
Theorem \ref{version_intro} establishes a new relation between time dependent
scattering theory (l.h.s.) and stationary scattering theory (r.h.s.) for a general
class of operators. We refer to the discussion in Section \ref{Interpretation} for
more information on these issues.

Let us now describe more precisely the content of this paper. In Section
\ref{Critical} we recall the necessary definitions from the theory of the conjugate
operator and define a critical set $\kappa(H)$ for the operator $H$. In the more
usual setup where $H=h(P)$ is a function of the momentum vector operator $P$ and
$\Phi$ is the position vector operator $Q$ in $\ltwo(\R^d)$, it is known that the
critical values of $h$
$$
\kappa_h:=\big\{\lambda\in\R\mid\exists\;\!x\in\R^d\hbox{ such that }
h(x)=\lambda\hbox{ and }h'(x)=0\big\}
$$
plays an important role (see \eg \cite[Sec.~7]{ABG}). Typically, the operator
$h(P)$ has bad spectral properties and bad propagation properties on $\kappa_h$.
For instance, one cannot obtain a simple Mourre estimate at these values. Such
phenomena also occur in the abstract setup. Since the operator $H$ is a priori not
a function of an auxiliary operator as $h(P)$, the derivative appearing in the
definition of $\kappa_h$ does not have a direct counterpart. However, the
identities $(\partial_jh)(P)=i[h(P),Q_j]$ suggest to define the set of critical
values $\kappa(H)$ in terms of the vector operator $H':=\big(i[H,\Phi_1],\ldots,i[H,\Phi_d]\big)$. This is the content of Definition \ref{surkappa}. In Lemma \ref{Heigen} and Theorem \ref{not_bad}, we show that
$\kappa(H)$ is closed, contains the set of eigenvalues of $H$, and that the
spectrum of $H$ in $\sigma(H)\setminus\kappa(H)$ is purely absolutely continuous.
The proof of the latter result relies on the construction, described in Section
\ref{SecContinuity}, of an appropriate conjugate operator for $H$.

In Section \ref{Averaged}, we recall some definitions in relation with the function
$f$ that appear in Theorem \ref{version_intro}. The function $R_f$ is introduced and
some of its properties are presented. Section \ref{Integral} is the core of the
paper and its most technical part. It contains the definition of $T_f$ and the proof
of the precise version of Theorem \ref{version_intro}. Suitable subspaces of $\H$ on
which the operators are well-defined and on which the equalities hold are
introduced.

An interpretation of our results is proposed in Section \ref{Interpretation}. The
relation with the theory of time operators is explained, and various cases are
presented. The importance of Theorem \ref{for_Schwartz} for the proof of the
existence of the quantum time delay and Eisenbud-Wigner Formula is also sketched.

In Section \ref{SecEx}, we show that our results apply to many operators $H$
appearing in physics and mathematics literature. Among other examples, we treat
Friedrichs Hamiltonians, Stark Hamiltonians, some Jacobi operators, the Dirac
operator, convolution operators on locally compact groups, pseudodifferential
operators, adjacency operators on graphs and direct integral operators. In each
case, we are able to exhibit a natural family of position operators $\Phi$
satisfying our assumptions. The diversity of the examples
covered by our theory make us strongly believe that Formula \eqref{Eq_prin} is of
natural character. Moreover it also suggests that the existence of time delay is
a very common feature of quantum scattering theory. We also point out that one
by-product of our study is an efficient algorithm for the choice of a conjugate
operator for a given self-adjoint operator $H$ (see Section \ref{SecContinuity}).
This allows us to obtain (or reobtain) non trivial spectral results for various
important classes of self-adjoint operators $H$.

As a final comment, we would like to emphasize that one of the main interest of our
study comes from the fact that we do not restrict ourselves to the standard position
operators $Q$ and to operators $H$ which are functions of $P$. Due to this
generality, we cannot rely on the usual canonical commutation relation of $Q$ and
$P$ and on the subjacent Fourier analysis. This explains the constant use of abstract
commutators methods throughout the paper.

\section{Critical values}\label{Critical}
\setcounter{equation}{0}

In this section, we recall some standard notions on the conjugate operator theory and introduce our general framework. The set of critical values is defined and some of
its properties are outlined. This subset of the spectrum of the operator under
investigation plays an essential role in the sequel.

We first recall some facts principally borrowed from \cite{ABG}. Let $H$ and $A$ be
two self-adjoint operators in a Hilbert space $\H$. Their respective domain are
denoted by $\dom(H)$ and $\dom(A)$, and for suitable $\omega\in\C$ we write
$R_\omega$ for $(H-\omega)^{-1}$. The operator $H$ is of class $C^1(A)$ if there
exists $\omega\in\C\setminus \sigma(H)$ such that the map
\begin{equation}\label{C1}
\R \ni t \mapsto \e^{-itA}R_\omega\e^{itA} \in \B(\H)
\end{equation}
is strongly differentiable. In that case, the quadratic form
$$
\dom(A)\ni\v\mapsto\langle A\v,R_\omega\v\rangle
-\langle R_\omega^*\v,A\v\rangle\in\C
$$
extends continuously to a bounded operator denoted by $[A,R_\omega]\in \B(\H)$. It
also follows from the $C^1(A)$-condition that $\dom(H)\cap\dom(A)$ is a core for $H$
and that the quadratic form
$
\dom(H)\cap\dom(A)\ni\v\mapsto\langle H\v,A\v\rangle-\langle A\v,H\v\rangle
$
is continuous in the topology of $\dom(H)$. This form extends then uniquely to a
continuous quadratic form $[H,A]$ on $\dom(H)$, which can be identified with a
continuous operator from $\dom(H)$ to $\dom(H)^*$. Finally, the following equality
holds:
\begin{equation}\label{visible}
[A,R_\omega]=R_\omega[H,A]R_\omega.
\end{equation}
It is also proved in \cite[Lemma~2]{GG99} that if $[H,A]\dom(H)\subset \H$, then
the unitary group $\{\e^{itA}\}_{t\in\R}$ preserves the domain of $H$, \ie $\e^{itA}\dom(H)\subset\dom(H)$ for all $t\in\R$.

We now extend this framework in two directions: in the number of conjugate
operators and in the degree of regularity with respect to these operators. So,
let us consider a family $\Phi\equiv(\Phi_1,\ldots,\Phi_d)$ of mutually
commuting self-adjoint operators in $\H$ (throughout the paper, we use the term
``commute" for operators commuting in the sense of \cite[Sec.~VIII.5]{RSI}).
Then we know from \cite[Sec. 6.5]{BS87} that any measurable function
$f\in\linf(\R^d)$ defines a bounded operator $f(\Phi)$ in $\H$. In particular,
the operator $\e^{ix\cdot\Phi}$, with $x\cdot\Phi\equiv\sum_{j=1}^dx_j\Phi_j$,
is unitary for each $x\in\R^d$. Note also that the conjugation
$$
C_x:\B(\H)\to\B(\H),\quad B\mapsto\e^{-ix\cdot\Phi}B\e^{ix\cdot\Phi}
$$
defines an automorphism of $\B(\H)$.

Within this framework, the operator $H$ is said to be of class $C^m(\Phi)$ for
$m=1,2,\ldots$ if there exists $\omega\in\C\setminus\sigma(H)$ such that the
map
\begin{equation}\label{Cm}
\R^d\ni x\mapsto\e^{-ix\cdot\Phi}R_\omega\e^{ix\cdot\Phi}\in\B(\H)
\end{equation}
is strongly of class $C^m$ in $\H$. One easily observes that if $H$ is of class
$C^m(\Phi)$, then the operator $H$ is of class $C^m(\Phi_j)$ for each $j$ (the
class $C^m(\Phi_j)$ being defined similarly).

\begin{Remark}
{\rm A bounded operator $S\in\B(\H)$ belongs to $C^1(A)$ if the map \eqref{C1}, with
$R_\omega$ replaced by $S$, is strongly differentiable. Similarly, $S\in\B(\H)$
belongs to $C^m(\Phi)$ if the map \eqref{Cm}, with $R_\omega$ replaced by $S$,
is strongly $C^m$.}
\end{Remark}

In the sequel, we assume that $H$ is regular with respect to unitary group
$\{\e^{ix\cdot \Phi}\}_{x\in\R^d}$ in the following sense.

\begin{Assumption}\label{chirimoya}
{\rm The operator $H$ is of class $C^3(\Phi)$. Furthermore, for each
$j\in\{1,\ldots,d\}$, the quadratic form $i[H,\Phi_j]$ on $\dom(H)$ defines an
essentially self-adjoint operator whose self-adjoint extension is denoted by
$H'_j$. Similarly, for each $k,\ell\in\{1,\ldots,d\}$, the quadratic form
$i[H'_j,\Phi_k]$ on $\dom(H'_j)$ defines an essentially self-adjoint operator
whose self-adjoint extension is denoted by $H''_{jk}$, and the quadratic form $i[H''_{jk},\Phi_\ell]$ on $\dom(H''_{jk})$ defines an essentially self-adjoint
operator whose self-adjoint extension is denoted by $H'''_{jk\ell}$.}
\end{Assumption}

This assumption implies the invariance of $\dom(H)$ under the action of the
unitary group $\{\e^{ix\cdot\Phi}\}_{x \in \R^d}$. Indeed, if the quadratic
form $i[H,\Phi_j]$ on $\dom(H)$ defines an essentially self-adjoint operator
in $\H$, it follows in particular that $\dom(H)\subset \dom(H'_j)$ and thus $i[H,\Phi_j]\dom(H)\equiv H'_j\dom(H)\subset\H$. It follows then from
\cite[Lemma~2]{GG99} that $\e^{it\Phi_j}\dom(H)\subset\dom(H)$ for all $t\in\R$.
In fact, one easily obtains that $\e^{it\Phi_j}\dom(H)=\dom(H)$, and since this
property holds for each $j$ one also has $\e^{ix\cdot\Phi}\dom(H)=\dom(H)$ for
all $x\in\R^d$. As a consequence, we obtain in particular that each self-adjoint
operator
\begin{equation}\label{H(x)}
H(x):=\e^{-ix\cdot \Phi}H\e^{ix\cdot \Phi}
\end{equation}
(with $H(0)=H$) has domain $\dom[H(x)]=\dom(H)$.

Similarly, the domains $\dom(H'_j)$ and $\dom(H''_{jk})$ are left invariant
by the action of the unitary group $\{\e^{ix\cdot \Phi}\}_{x\in\R^d}$, and
the operators $H'_j(x):=\e^{-ix\cdot\Phi}H'_j\e^{ix\cdot\Phi}$ and $H''_{jk}(x):=\e^{-ix\cdot\Phi}H''_{jk}\e^{ix\cdot\Phi}$ are self-adjoint
operators with domains $\dom(H'_j)$ and $\dom(H''_{jk})$ respectively.

Our second main assumption concerns the family of operators $H(x)$.

\begin{Assumption}\label{commute}
{\rm The operators $\{H(x)\}_{x\in \R^d}$ mutually commute.}
\end{Assumption}

Using the fact that the map $\R^d\ni x\mapsto C_x\in\Aut[\B(\H)]$ is a group
morphism, one easily shows that Assumption \ref{commute} is equivalent the
commutativity of each $H(x)$ with $H$.
Furthermore, Assumptions \ref{chirimoya} and
\ref{commute} imply additional commutation relations:

\begin{Lemma}\label{undos}
The operators $H(x)$, $H'_j(y)$, $H''_{k\ell}(z)$  mutually commute for each
$j,k,\ell\in\{1,\ldots,d\}$ and each $x,y,z\in\R^d$.
\end{Lemma}

\begin{proof}
Let $\omega\in\C\setminus\R$, $x,y,z\in\R^d$, $j,k,\ell,m\in\{1,\ldots,d\}$,
and set $R(x):=[H(x)-\omega]^{-1}$, $R'_j(x):=[H'_j(x)-\omega]^{-1}$ and $R''_{jk}(x):=[H''_{jk}(x)-\omega]^{-1}$. By assumption, one has the equality
$$
\textstyle
R(x)\;\!\frac{R(\varepsilon e_j)-R(0)}\varepsilon
=\frac{R(\varepsilon e_j)-R(0)}\varepsilon\;\!R(x)
$$
for each $\varepsilon\in\R\setminus\{0\}$. Taking the strong limit as
$\varepsilon\to0$, and using \eqref{visible} and Assumption \ref{commute},
one obtains
$$
R(0)\[R(x)H'_j-H'_jR(x)\]R(0)=0.
$$
Since the resolvent $R(0)$ on the left is injective, this implies that
$R(x)H'_j-H'_jR(x)=0$ on $\dom(H)$. Furthermore, since $\dom(H)$ is a core
for $H'_j$ the last equality can be extended to $\dom(H'_j)$. Finally, by
multiplying the equation
$$
R(x)=R(x)\big(H'_j-\omega\big)R'_j(0)=\big(H'_j-\omega\big)R(x)R'_j(0)
$$
on the left by $R'_j(0)$, one gets $R'_j(0)R(x)=R(x)R'_j(0)$. Using the
morphism property of the map $\R^d\ni x\mapsto C_x\in\Aut[\B(\H)]$, one
infers from this that $H(x)$ and $H'_j(y)$ commute.

A similar argument leads to the commutativity of the operators $H'_j(x)$ and
$H'_k(y)$ by considering the operators
$R'_j(x)\frac{R(\varepsilon e_k)-R(0)}\varepsilon$ and
$\frac{R(\varepsilon e_k)-R(0)}\varepsilon R'_j(x)$. The commutativity of
$H(x)$ and $H''_{jk}(z)$ is obtained by considering the operators
$R(x)\frac{R'_j(\varepsilon e_k)-R'_j(0)}\varepsilon$ and
$\frac{R'_j(\varepsilon e_k)-R'_j(0)}\varepsilon R(x)$, and the commutativity
of $H'_j(y)$ and $H''_{k\ell}(z)$ by considering the operators
$R'_j(y)\frac{R'_k(\varepsilon e_\ell)-R'_k(0)}\varepsilon$ and
$\frac{R'_k(\varepsilon e_\ell)-R'_k(0)}\varepsilon R'_j(y)$. Finally, the
commutation between $H''_{jk}(x)$ and $H''_{\ell m}(y)$ is obtained by
considering the operators
$R''_{jk}(x)\frac{R'_\ell(\varepsilon e_m)-R'_{\ell}(0)}\varepsilon$ and
$\frac{R'_\ell(\varepsilon e_m)-R'_\ell(0)}\varepsilon R''_{jk}(x)$.  Details
are left to the reader.
\end{proof}

For simplicity, we write $H'$ for the vector operator
$(H'_1,\ldots,H'_d)$, and define for each measurable function $f:\R^d\to\C$
the operator $f(H')$ by using the $d$-variables functional calculus. The
symbol $E^H(\cdot)$ denotes the spectral measure of $H$.

\begin{Definition}\label{surkappa}
{\rm A number $\lambda\in\R$ is called a regular value of $H$ if there exists
$\delta>0$ such that
\begin{equation}\label{condition}
\lim_{\varepsilon\searrow0}\big\|\big[(H')^2+\varepsilon\big]^{-1}
E^H\big((\lambda-\delta,\lambda+\delta)\big)\big\|<\infty.
\end{equation}
A number $\lambda\in\R$ that is not a regular value of $H$ is called a
critical value of $H$. We denote by $\kappa(H)$ the set of critical values of
$H$.}
\end{Definition}

From now on, we shall use the shorter notation $E^H(\lambda;\delta)$ for $E^H\big((\lambda-\delta,\lambda+\delta)\big)$. In the next lemma we put into
evidence some useful properties of the set $\kappa(H)$.

\begin{Lemma}\label{Heigen}
Let Assumptions \ref{chirimoya} and \ref{commute} be verified. Then the set
$\kappa(H)$ possesses the following properties:
\begin{enumerate}
\item[(a)] $\kappa(H)$ is closed.
\item[(b)] $\kappa(H)$ contains the set of eigenvalues of $H$.
\item[(c)] The limit
$
\lim_{\varepsilon\searrow0}\big\|\big[(H')^2+\varepsilon\big]^{-1}E^H(J)\big\|
$
is finite for each compact set $J \subset \R\setminus \kappa(H)$.
\item[(d)] For each compact set $J\subset\R\setminus\kappa(H)$, there exists a
compact set $U\subset(0,\infty)$ such that $E^H(J)=E^{|H'|}(U)E^H(J)$.
\end{enumerate}
\end{Lemma}

\begin{proof}
(a) Let $\lambda_0$ be a regular value for $H$, \ie there exists $\delta_0>0$
such that \eqref{condition} holds with $\delta$ replaced by $\delta_0$. Let
$\lambda\in(\lambda_0-\delta_0,\lambda_0+\delta_0)$ and let $\delta >0$ such
that
$$
(\lambda-\delta,\lambda+\delta)\subset(\lambda_0-\delta_0,\lambda_0+\delta_0).
$$
Then, since $E^H(\lambda;\delta)=E^H(\lambda_0;\delta_0)E^H(\lambda;\delta)$,
one has
\begin{equation*}
\lim_{\varepsilon\searrow0}
\big\|\big[(H')^2+\varepsilon\big]^{-1}E^H(\lambda;\delta)\big\|
\leq\lim_{\varepsilon\searrow0}
\big\|\big[(H')^2+\varepsilon\big]^{-1}E^H(\lambda_0;\delta_0)\big\|
 <\infty.
\end{equation*}
But this means exactly that $\lambda$ is a regular value for any $\lambda\in(\lambda_0-\delta_0,\lambda_0+\delta_0)$. So the set of regular
values is open, and $\kappa(H)$ is closed.

(b) Let $\lambda\in\R$ be an eigenvalue of $H$, and let $\v_\lambda$ be an
associated eigenvector with norm one. Since $H$ is of class $C^1(\Phi_j)$
for each $j$, we know from the Virial theorem \cite[Prop.~7.2.10]{ABG} that $E^H(\{\lambda\})H_j'E^H(\{\lambda\})=0$ for each $j$. This, together with
Lemma \ref{undos}, implies that
$$
E^H(\{\lambda\})\big[(H')^2+\varepsilon\big]^{-1}E^H(\{\lambda\})
=\varepsilon^{-1}E^H(\{\lambda\})
$$
for each $\varepsilon>0$. In particular, we obtain for each $\delta>0$ the
equalities
$$
\big[(H')^2+\varepsilon\big]^{-1}E^H(\lambda;\delta)\v_\lambda
=E^H(\{\lambda\})\big[(H')^2+\varepsilon\big]^{-1}E^H(\{\lambda\})
\v_\lambda=\varepsilon^{-1}\v_\lambda,
$$
and
\begin{equation*}
\lim_{\varepsilon\searrow0}
\big\|\big[(H')^2+\varepsilon\big]^{-1}E^H(\lambda;\delta)\big\|
\ge\lim_{\varepsilon\searrow0}
\big\|\big[(H')^2+\varepsilon\big]^{-1}E^H(\lambda;\delta)\v_\lambda\big\|
=\lim_{\varepsilon\searrow0}\varepsilon^{-1}\|\v_\lambda\|
=\infty.
\end{equation*}
Since $\delta$ has been chosen arbitrarily, this implies that $\lambda$ is
not a regular value of $H$.

(c) This follows easily by using a compacity argument.

(d) Let us concentrate first on the lower bound of $U$. Clearly, if $|H'|$ is
strictly positive, then $U$ can be chosen in $(0,\infty)$ and thus is bounded
from below by a strictly positive number. So assume now that $|H'|$ is not
strictly positive, that is $0 \in \sigma(|H'|)$. By absurd, suppose  that $U$
is not bounded from below by a strictly positive number, \ie there does not
exist $a>0$ such that $U \subset (a,\infty)$. Then for $n=1,2,\dots$, there
exists $\psi_n\in \H$ such that $E^{|H'|}\big([0,1/n)\big)E^H(J)\psi_n\neq0$,
and the vectors
$$
\v_n:=\frac{E^{|H'|}\big([0,1/n)\big)E^H(J)\psi_n}
{\|E^{|H'|}\big([0,1/n)\big)E^H(J)\psi_n\|}
$$
satisfy $\|\v_n\|=1$, and $E^H(J)\v_n=E^{|H'|}\big([0,1/n)\big)\v_n=\v_n$. It
follows by point (c) that
\begin{align*}
{\rm Const.}\ge\lim_{\varepsilon\searrow0}
\big\|\big[(H')^2+\varepsilon\big]^{-1}E^H(J)\big\|
&\ge\lim_{\varepsilon\searrow0}
\big\|\big[(H')^2+\varepsilon\big]^{-1}E^H(J)\v_n\big\|\\
&=\lim_{\varepsilon\searrow0}
\big\|\big[(H')^2+\varepsilon\big]^{-1}E^{|H'|}\big([0,1/n)\big)\v_n\big\|\\
&\ge\lim_{\varepsilon\searrow0}\big(n^{-2}+\varepsilon\big)^{-1}\|\v_n\|\\
&=n^2,
\end{align*}
which leads to a contradiction when $n\to\infty$.

Let us now concentrate on the upper bound of $U$. Clearly, if $|H'|$ is a
bounded operator, one can choose a bounded subset $U$ of $\R$ and thus $U$ is
upper bounded. So assume now that $|H'|$ is not a bounded operator. By absurd,
suppose that $U$ is not bounded from above, \ie there does not exist $b<\infty$
such that $U \subset (0,b)$. Then for $n=1,2,\dots$, there exists $\psi_n\in\H$
such that $E^{|H'|}\big([n,\infty)\big)E^H(J)\psi_n\neq0$, and the vectors
$$
\v_n:=\frac{E^{|H'|}\big([n,\infty)\big)E^H(J)\psi_n}
{\|E^{|H'|}\big([n,\infty)\big)E^H(J)\psi_n\|}
$$
satisfy $\|\v_n\|=1$, and $E^H(J)\v_n=E^{|H'|}\big([n,\infty)\big)\v_n=\v_n$.
It follows by Assumption \ref{chirimoya} and Lemma \ref{undos} that
$|H'|\;\!E^H(J)$ is a bounded operator, and
\begin{equation*}
{\rm Const.}\ge\big\||H'|\;\!E^H(J)\big\|
\ge\big\||H'|\;\!E^H(J)\v_n\big\|
=\big\||H'|\;\!E^{|H'|}\big([n,\infty)\big)\v_n\big\|
\ge n\;\!\|\v_n\|\,
\end{equation*}
which leads to a contradiction when $n\to\infty$.
\end{proof}

\section{Locally smooth operators and absolute continuity}\label{SecContinuity}
\setcounter{equation}{0}

In this section we exhibit a large class of locally $H$-smooth operators. We also show
that the operator $H$ is purely absolutely continuous in $\sigma(H)\setminus\kappa(H)$. These results are obtained by using commutators methods as presented in \cite{ABG}.

In order to motivate our choice of conjugate operator for $H$, we present first a
formal calculation. Let $A_\eta$ be given by
$$
A_\eta:=\12\big\{\eta(H)H'\cdot\Phi+\Phi\cdot H'\eta(H)\big\},
$$
where $\eta$ is some real function with a sufficiently rapid decrease to $0$ at
infinity. Then $A_\eta$ satisfies with $H$ the commutation relation
\begin{equation*}
\textstyle i[H,A_\eta]
=\frac i2\sum_{j=1}^d \big\{\eta(H)H'_j\;\![H,\Phi_j]+[H,\Phi_j]\;\!H'_j\eta(H)\big\}
=(H')^2\eta(H),
\end{equation*}
which provides (in a sense to be specified) a Mourre estimate. So, in the sequel,
one only has to justify these formal manipulations and to determinate an appropriate
function $\eta$.

First of all, one observes that for each $j\in\{1,\ldots,d\}$ and each $\omega\in\C\setminus\sigma(H)$ the operator $H'_jR_\omega\equiv H'_j(H-\omega)^{-1}$
is a bounded operator. Indeed, one has $(H-\omega)^{-1}\H= \dom(H)\subset\dom(H'_j)$
by Assumption \ref{chirimoya}. In the following lemmas, Assumptions \ref{commute} and \ref{chirimoya} are tacitly assumed, and we set $\langle x\rangle:=(1+x^2)^{1/2}$ for
any $x\in\R^n$.

\begin{Lemma}\label{OnPi}
\begin{enumerate}
\item[(a)] For each $j,k\in\{1,\ldots,d\}$ and each
$\gamma,\omega\in\C\setminus\sigma(H)$, the bounded operator $R_\gamma H'_jR_\omega$
belongs to $C^1(\Phi_k)$.
\item[(b)] For each $j,k\in\{1,\ldots,d\}$ the bounded self-adjoint operator
$\langle H\rangle^{-2}H'_j\langle H\rangle^{-2}$ belongs to $C^1(\Phi_k)$.
\item[(c)] For each $j,k,\ell\in\{1,\ldots,d\}$, the bounded self-adjoint operator $i\big[\langle H\rangle^{-2}H'_j\langle H\rangle^{-2},\Phi_k\big]$ belongs to $C^1(\Phi_\ell)$.
\end{enumerate}
\end{Lemma}

\begin{proof}
Due to Assumption \ref{chirimoya} one has for each $\v\in\dom(\Phi_k)$
\begin{align*}
&\big\langle\Phi_k\v,R_\gamma H'_jR_\omega\v\big\rangle
-\big\langle R_{\bar\omega}H'_jR_{\bar \gamma}\v,\Phi_k\v\big\rangle\\
&=\big\langle\Phi_k\v,R_\gamma H'_jR_\omega\v\big\rangle
-\big\langle\Phi_kR_{\bar \gamma}\v,H'_jR_\omega\v\big\rangle
+\big\langle\Phi_kR_{\bar \gamma}\v,H'_jR_\omega\v\big\rangle
-\big\langle R_{\bar\omega}H'_jR_{\bar \gamma}\v,\Phi_k\v\big\rangle\\
&=\big\langle[R_{\bar \gamma},\Phi_k]\v,H'_jR_\omega\v\big\rangle
+\big\langle\Phi_kR_{\bar \gamma}\v,H'_jR_\omega\v\big\rangle
-\big\langle H'_jR_{\bar \gamma}\v,\Phi_kR_\omega\v\big\rangle\\
&\quad+\big\langle H'_jR_{\bar \gamma}\v,\Phi_kR_\omega\v\big\rangle
-\big\langle R_{\bar \omega}H'_jR_{\bar \gamma}\v,\Phi_k\v\big\rangle\\
&=\big\langle[R_{\bar \gamma},\Phi_k]\v,H'_jR_\omega\v\big\rangle
+\big\langle[H'_j,\Phi_k]R_{\bar \gamma}\v,R_\omega\v\big\rangle
+\big\langle H'_jR_{\bar \gamma}\v,[\Phi_k,R_\omega]\v\big\rangle.
\end{align*}
This implies that there exists ${\textsc c}<\infty$ such that
$$
\big|\big\langle\Phi_k\v,R_\gamma H'_jR_\omega\v\big\rangle
-\big\langle R_{\bar\omega}H'_jR_{\bar \gamma}\v,\Phi_k\v\big\rangle\big|
\leq{\textsc c}\,\|\v\|^2.
$$
for each $\v\in\dom(\Phi_k)$, and thus the first statement follows from \cite[Lem.~6.2.9]{ABG}.

For the second statement, since $\langle H \rangle^{-2}=R_{-i}R_i$, the operator
$\langle H\rangle^{-2}H'_j\langle H\rangle^{-2}$ is clearly bounded and self-adjoint.
Furthermore, by observing that
\begin{equation*}
\langle H\rangle^{-2}H'_j\langle H\rangle^{-2}= R_i\big(R_{-i}H'_jR_i\big)R_{-i}
\end{equation*}
one concludes from (a) that $\langle H\rangle^{-2}H'_j\langle H\rangle^{-2}$ is the
product of three operators belonging to $C^1(\Phi_k)$, and thus belongs to
$C^1(\Phi_k)$ due to \cite[Prop.~5.1.5]{ABG}.

For the last statement, one gets by taking Lemma \ref{undos} into account
\begin{equation*}
i\big[\langle H \rangle^{-2}H'_j \langle H \rangle^{-2},\Phi_k\big]
=-2(R_iH_k'R_i)(R_{-i}H_j'R_{-i})(R_i+R_{-i})
+\langle H\rangle^{-2}H''_{jk}\langle H\rangle^{-2}.
\end{equation*}
The first term is a product of operators which belong to $C^1(\Phi_\ell)$, and thus
it belongs to $C^1(\Phi_\ell)$. For the second term, a calculation similar to the one
presented for the statement (a) using Assumption \ref{chirimoya} shows that this term
also belongs to $C^1(\Phi_\ell)$, and so the claim is proved.
\end{proof}

We can now give a precise definition of the conjugate operator $A$ we will
use, and prove its self-adjointness. For that purpose, we consider the family
\begin{equation*}
\Pi_j:=\<H\>^{-2}H'_j\<H\>^{-2},\qquad j=1,\ldots,d,
\end{equation*}
of mutually commuting bounded self-adjoint operators, and we write
$\Pi:=(\Pi_1,\ldots,\Pi_d)$ for the associated vector operator. Due to Lemma
\ref{OnPi}.(b), each operator $\Pi_j$ belongs to $C^1(\Phi_k)$. Therefore the
operator
$$
A:=\12\big(\Pi\cdot\Phi+\Phi\cdot\Pi\big)
$$
is well-defined and symmetric on $\bigcap_{j=1}^d \dom(\Phi_j)$. For the next
lemma, we note that this set contains the domain $\dom(\Phi^2)$ of $\Phi^2$.

\begin{Lemma}
The operator $A$ is essentially self-adjoint on $\dom(\Phi^2)$.
\end{Lemma}

\begin{proof}
We use the criterion of essential self-adjointness \cite[Thm.~X.37]{RSII}.

Given $a>1$, we define the self-adjoint operator $N:=\Phi^2+\Pi^2+a$ with domain $\dom(N)\equiv\dom(\Phi^2)$ and observe that in the form sense on $\dom(N)$ one
has
\begin{align*}
N^2&=\Phi^4+\Pi^4+a^2+2a\Phi^2+2a\Pi^2+\Phi^2\Pi^2+\Pi^2\Phi^2\\
&=\Phi^4+\Pi^4+a^2+2a\Phi^2+2a\Pi^2
+\sum_{j,k}\big\{\Phi_j\Pi_k^2\Phi_j+\Pi_k\Phi_j^2\Pi_k\big\}+ R
\end{align*}
with
$
R:=\sum_{j,k}\big\{\Pi_k[\Pi_k,\Phi_j]\Phi_j+\Phi_j[\Phi_j,\Pi_k]\Pi_k
+[\Pi_k,\Phi_j]^2\big\}$.
Now, the following inequality holds
$$
\sum_{j,k}\big\{\Pi_k[\Pi_k,\Phi_j]\Phi_j+ \Phi_j[\Phi_j,\Pi_k]\Pi_k\big\}
\geq-d\Phi^2-\sum_{j,k}\big|\Pi_k[\Pi_k,\Phi_j]\big|^2.
$$
Thus there exists $c>0$ such that $R\geq-d\Phi^2-c$. Altogether, we have shown
that in the form sense on $\dom(N)$
$$
N^2\geq\Phi^4+\Pi^4+(a^2-c)+(2a-d)\Phi^2+2a\Pi^2
+\sum_{j,k}\big\{\Phi_j\Pi_k^2\Phi_j+\Pi_k\Phi_j^2\Pi_k\big\},
$$
where the r.h.s. is a sum of positive terms for $a$ large enough. In particular,
one has for $\v\in\dom(N)$
$$
\|N\v\|^2\ge\big\|\Pi_j\Phi_j\v\big\|^2+\big\|\Phi_j\Pi_j\v\big\|^2,
$$
which implies that
$$
\|A\v\|\leq\12\sum_j\big\{\big\|\Pi_j\Phi_j\v\big\|
+\big\|\Phi_j\Pi_j\v\big\|\big\}
\leq d\,\|N\v\|\,.
$$

It remains to estimate the commutator $[A,N]$. In the form sense on $\dom(N)$, one has
\begin{align*}
2[A,N]&=\sum_{j,k}\big\{[\Pi_j,\Phi_k]\Phi_j\Phi_k+\Phi_k[\Pi_j,\Phi_k]\Phi_j
+\Phi_j[\Pi_j,\Phi_k]\Phi_k+\Phi_j\Phi_k[\Pi_j,\Phi_k]\\
&\qquad+\Pi_j[\Phi_j,\Pi_k]\Pi_k+\Pi_j\Pi_k[\Phi_j,\Pi_k]+[\Phi_j,\Pi_k]\Pi_j\Pi_k
+\Pi_k[\Phi_j,\Pi_k]\Pi_j\big\}.
\end{align*}
The last four terms are bounded. For the other terms, Lemma \ref{OnPi}.(c), together
with the bound
$$
|\langle\Phi_j\v,B\Phi_k\rangle|
\le\|B\|\,\langle\v,\Phi^2\v\rangle
\le\|B\|\,\langle\v,N\v\rangle,
\qquad\v\in\dom(N),~B\in\B(\H),
$$
leads to the desired estimate, \ie
$\langle\v,[A,N]\v\rangle\leq{\rm Const.}\;\!\langle\v,N\v\rangle$.
\end{proof}

\begin{Lemma}\label{C2}
The operator $H$ is of class $C^2(A)$ and the sesquilinear form $i[H,A]$ on
$\dom(H)$ extends to the bounded positive operator
$\langle H\rangle^{-2}(H')^2\langle H\rangle^{-2}$.
\end{Lemma}

\begin{proof}
One has for each $\v\in\dom(\Phi^2)$ and each $\omega\in\C\setminus\sigma(H)$
\begin{align}\label{exemple}
2\big\{\big\langle R_{\bar \omega}\v,A\v\big\rangle
-\big\langle A\v,R_\omega\v\big\rangle\big\}
&=\sum_j\big\{\big\langle R_{\bar \omega}\v,\big(\Pi_j\Phi_j+\Phi_j\Pi_j\big)\v\big\rangle
-\big\langle\big(\Pi_j\Phi_j+\Phi_j\Pi_j\big)\v,R_\omega\v\big\rangle\big\}\nonumber\\
&=\sum_j\big\{\big\langle\Pi_j\v,\[R_\omega,\Phi_j\]\v\big\rangle
+\big\langle\[\Phi_j,R_{\bar \omega}\]\v,\Pi_j\v\big\rangle\big\}.
\end{align}
Since all operators in the last equality are bounded and since $\dom(\Phi^2)$ is a
core for $A$, this implies that $H$ is of class $C^1(A)$ \cite[Lem.~6.2.9]{ABG}.

Now observe that the following equalities hold on $\H$
\begin{equation*}
\textstyle i[R_\omega,A]
=\frac i2\sum_j\big\{\Pi_j[R_\omega,\Phi_j]+[R_\omega,\Phi_j]\Pi_j\big\}
=-R_\omega\<H\>^{-2}(H')^2\<H\>^{-2}R_\omega.
\end{equation*}
Therefore the sesquilinear form $i[H,A]$ on $\dom(H)$ extends to the bounded
positive operator $\<H\>^{-2}(H')^2\<H\>^{-2}$. Finally, the operator $i[R_\omega,A]$
can be written as a product of factors in $C^1(\Phi_\ell)$ for each $\ell$, namely
$$
\textstyle
i[R_\omega,A]=-\sum_jR_\omega\(R_{-i}H'_jR_i\)\(R_{-i}H'_jR_i\)R_\omega.
$$
So $i[R_\omega,A]$ also belongs to $C^1(\Phi_\ell)$ for each $\ell$, and thus a
calculation similar to the one of \eqref{exemple} shows that $i[R_\omega,A]$ belongs to
$C^1(A)$. This implies that $H$ is of class $C^2(A)$.
\end{proof}

\begin{Definition}\label{surkappaA}
{\rm A number $\lambda\in\R$ is called a $A$-regular value of $H$ if there exist
numbers $a,\delta>0$ such that
$(H')^2E^H(\lambda;\delta)\geq a\;\!E^H(\lambda;\delta)$. The complement of
this set in $\R$ is denoted by $\kappa^A(H)$.}
\end{Definition}

The set of $A$-regular values corresponds to the Mourre set with respect to
$A$. Indeed, if $\lambda$ is a $A$-regular value, then
$(H')^2E^H(\lambda;\delta)\geq a\;\!E^H(\lambda;\delta)$ for some $a,\delta>0$
and
$$
E^H(\lambda;\delta)i[H,A]E^H(\lambda;\delta)
=E^H(\lambda;\delta)\<H\>^{-2}(H')^2\<H\>^{-2}E^H(\lambda;\delta)
\ge a'E^H(\lambda;\delta),
$$
where $a':=a\cdot\inf_{\mu\in(\lambda-\delta,\lambda+\delta)}\langle\mu\rangle^{-4}$.
In the framework of Mourre theory, this means that the operator $A$ is strictly
conjugate to $H$ at the point $\lambda$ \cite[Sec.~7.2.2]{ABG}.

\begin{Lemma}\label{egalite_k}
The sets $\kappa(H)$ and $\kappa^A(H)$ are equal.
\end{Lemma}

\begin{proof}
Let $\lambda$ be a $A$-regular value of $H$. Then there exist $a,\delta>0$ such that
$$
E^H(\lambda;\delta)\le a^{-1}(H')^2E^H(\lambda;\delta),
$$
and we obtain for $\varepsilon >0$:
\begin{align*}
\big\|\big[(H')^2+\varepsilon\big]^{-1}E^H(\lambda;\delta)\big\|^2
&=\sup_{\v\in\H,\,\|\v\|=1}
\big\langle\big[(H')^2+\varepsilon\big]^{-1}\v,E^H(\lambda;\delta)
\big[(H')^2+\varepsilon^2\big]^{-1}\v\big\rangle\\
&\le a^{-2}\sup_{\v\in\H,\,\|\v\|=1}
\big\langle\big[(H')^2+\varepsilon\big]^{-1}\v,E^H(\lambda;\delta)(H')^4
[(H')^2+\varepsilon]^{-1}\v\big\rangle\\
&\le a^{-2}
\big\|(H')^2[(H')^2+\varepsilon]^{-1}\big\|^2\\
&\leq a^{-2},
\end{align*}
which implies, by taking the limit $\lim_{\varepsilon\searrow0}$, that $\lambda$ is
a regular value.

Now, let $\lambda$ be a regular value of $H$. Then there exists $\delta>0$ such that
\begin{align}
{\rm Const.}\ge\lim_{\varepsilon\searrow0}
\big\|\big[(H')^2+\varepsilon\big]^{-1}E^H(\lambda;\delta)\big\|
&=\lim_{\varepsilon\searrow0}
\big\|E^H(\lambda;\delta)\big[(H')^2E^H(\lambda;\delta)+\varepsilon\big]^{-1}E^H(\lambda;\delta)\big\|\nonumber\\
&=\lim_{\varepsilon\searrow0}
\big\|\big[(H')^2E^H(\lambda;\delta)+\varepsilon\big]^{-1}\big\|
_{\B(\H_{\lambda,\delta})},
\label{plofplof}
\end{align}
where $\H_{\lambda,\delta}:=E^H(\lambda;\delta)\H$. But we have
$$
\big\|\big[(H')^2E^H(\lambda;\delta)+\varepsilon\big]^{-1}\big\|_{\B(\H_{\lambda,\delta})}
=(a+\varepsilon)^{-1},
$$
where the number $a\ge0$ is the infimum of the spectrum of $(H')^2E^H(\lambda;\delta)$,
considered as an operator in $\H_{\lambda,\delta}$. Therefore, Formula \eqref{plofplof}
entails the bound $a^{-1}\le{\rm Const.}$, which implies that $a>0$. In consequence, the
operator $(H')^2E^H(\lambda;\delta)$ is strictly positive in $\H_{\lambda,\delta}$,
namely,
$$
(H')^2E^H(\lambda;\delta)\ge aE^H(\lambda;\delta)
$$
with $a>0$. This implies that $\lambda$ is a $A$-regular value of $H$, and $\kappa(H)$
is equal to $\kappa^A(H)$.
\end{proof}

We shall now state our main result on the nature of the spectrum of $H$, and
exhibit a class of locally $H$-smooth operators. The space
$\big(\dom(A),\H\big)_{1/2,1}$, defined by real interpolation
\cite[Sec.~3.4.1]{ABG}, is denoted by $\KK$. Since for each $j\in\{1,\ldots,d\}$
the operator $\Pi_j$ belongs to $C^1(\Phi_j)$, we have
$\dom(\<\Phi\>)\subset\dom(A)$, and it follows from \cite[Thm.~2.6.3]{ABG}
and \cite[Thm.~3.4.3.(a)]{ABG} that for $s>1/2$ the continuous embeddings hold:
\begin{equation}\label{truite}
\dom(\<\Phi\>^s)\subset\KK\subset\H\subset\KK^*\subset\dom(\<\Phi\>^{-s}).
\end{equation}
The symbol $\C_\pm$ stands for the half-plane
$\C_\pm:=\{\omega\in\C\mid\pm\im(\omega)>0\}$.

\begin{Theorem}\label{not_bad}
Let $H$ satisfy Assumptions \ref{chirimoya} and \ref{commute}. Then,
\begin{enumerate}
\item[(a)] the spectrum of $H$ in $\sigma(H)\setminus\kappa(H)$ is purely
absolutely continuous,
\item[(b)] each operator $T\in\B\big(\dom(\<\Phi\>^{-s}),\H\big)$, with
$s>1/2$, is locally $H$-smooth on $\R\setminus\kappa(H)$.
\end{enumerate}
\end{Theorem}

\begin{proof}
(a) This is a direct application of \cite[Thm.~0.1]{Sah97} which takes Lemmas
\ref{C2} and \ref{egalite_k} into account.

(b) We know from \cite[Thm.~0.1]{Sah97} that the map
$\omega\mapsto R_\omega\in\B(\KK,\KK^*)$, which is holomorphic on the half-plane
$\C_\pm$, extends to a weak$^*$-continuous function on
$\C_\pm\cup\{\R\setminus\kappa(H)\}$. Now, consider $T\in\B(\KK^*,\H)$. Then one
has $T^*\in\B(\H,\KK)$, and it follows from the above continuity that for each
compact subset $J\subset\R\setminus\kappa(H)$ there exists a constant
$\textsc c\ge0$ such that for all $\omega\in\C$ with $\re(\omega)\in J$ and
$\im(\omega)\in(0,1)$ one has
$$
\|TR_\omega T^*\|+\|TR_{\bar\omega}T^*\|\leq\textsc c.
$$
A fortiori, one also has $\sup_\omega\|T(R_\omega-R_{\bar\omega})T^*\|\leq\textsc c$,
where the supremum is taken over the same set of complex numbers. This last property
is equivalent to the local $H$-smoothness of $T$ on $\R\setminus\kappa(H)$. The
claim is then obtained by using the last embedding of \eqref{truite}.
\end{proof}

\section{Averaged localisation functions}\label{Averaged}
\setcounter{equation}{0}

In this section we recall some properties of a class of averaged localisation
functions which appears naturally when dealing with quantum scattering theory. These
functions, which are denoted $R_f$, are constructed in terms of functions
$f\in\linf(\R^d)$ of localisation around the origin $0$ of $\R^d$. They were already
used, in one form or another, in \cite{GT07}, \cite{Tie08}, and \cite{Tie09_3}.

\begin{Assumption}\label{assumption_f}
{\rm The function $f\in\linf(\R^d)$ satisfies the following conditions:
\begin{enumerate}
\item[(i)] There exists $\rho>0$ such that
$|f(x)|\le{\rm Const.}\,\langle x\rangle^{-\rho}$ for a.e. $x\in \R^d$.
\item[(ii)] $f=1$ on a neighbourhood of~~$0$.
\end{enumerate}}
\end{Assumption}

It is clear that $\slim_{r\to\infty}f(\Phi/r)=1$ if $f$ satisfies Assumption
\ref{assumption_f}. Furthermore, one has for each $x\in \R^d\setminus\{0\}$
$$
\left|\int_0^\infty\frac{\d\mu}\mu\[f(\mu x)-\chi_{[0,1]}(\mu)\]\right|
\le\int_0^1\frac{\d\mu}\mu\,|f(\mu x)-1|
+{\rm Const.}\int_1^{+\infty}\d\mu\,\mu^{-(1+\rho)}
<\infty,
$$
where $\chi_{[0,1]}$ denotes the characteristic function for the interval $[0,1]$.
Therefore the function $R_f:\R^d\setminus\{0\}\to\C$ given by
$$
R_f(x):=\int_0^{+\infty}\frac{\d\mu}\mu\[f(\mu x)-\chi_{[0,1]}(\mu)\]
$$
is well-defined. If $\R^*_+:=(0,\infty)$, endowed with the multiplication,
is seen as a Lie group with Haar measure $\frac{\d\mu}\mu$, then $R_f$ is the
renormalised average of $f$ with respect to the (dilation) action of $\R^*_+$
on $\R^d$.

In the next lemma we recall some differentiability and homogeneity properties
of $R_f$. We also give the explicit form of $R_f$ when $f$ is a radial function.
The reader is referred to \cite[Sec. 2]{Tie09_3} for proofs and details. The
symbol $\S(\R^d)$ stands for the Schwartz space on $\R^d$.

\begin{Lemma}\label{function_R}
Let $f$ satisfy Assumption \ref{assumption_f}.
\begin{enumerate}
\item[(a)] Assume that $(\partial_jf)(x)$ exists for all $j\in\{1,\ldots,d\}$ and
$x\in\R^d$, and suppose that there exists some $\rho>0$ such that
$|(\partial_jf)(x)|\le{\rm Const.}\<x\>^{-(1+\rho)}$ for each $x\in\R^d$. Then
$R_f$ is differentiable on $\R^d\setminus\{0\}$, and its derivative is given by
\begin{equation*}
R_f'(x)=\int_0^\infty\d\mu\,f'(\mu x).
\end{equation*}
In particular, if $f\in\S(\R^d)$ then $R_f$ belongs to $C^\infty(\R^d\setminus\{0\})$.
\item[(b)] Assume that $R_f$ belongs to $C^m(\R^d\setminus\{0\})$ for some
$m\ge1$. Then one has for each $x\in \R^d\setminus\{0\}$ and $t>0$ the homogeneity
properties
\begin{align}
x\cdot R_f'(x)&=-1,\label{minusone}\\
t^{|\alpha|}(\partial^\alpha R_f)(tx)
&=(\partial^\alpha R_f)(x),\label{R_f_alpha}
\end{align}
where $\alpha\in\N^d$ is a multi-index with $1\le|\alpha|\le m$.
\item[(c)] Assume that $f$ is radial, \ie there exists $f_0\in\linf(\R)$ such that $f(x)=f_0(|x|)$ for a.e. $x\in \R^d$. Then $R_f$ belongs to
$C^\infty(\R^d\setminus\{0\})$, and $R_f'(x)=-x^{-2}x$.
\end{enumerate}
\end{Lemma}

Obviously, one can show as in Lemma \ref{function_R}.(a) that $R_f$ is of class $C^m(\R^d\setminus\{0\})$ if one has for each $\alpha\in \N^d$ with $|\alpha|\le m$
that $(\partial^\alpha f)(x)$ exists and that
$|(\partial^\alpha f)(x)|\le{\rm Const.}\<x\>^{-(|\alpha|+\rho)}$ for some $\rho>0$. However, this is not a necessary condition. In some cases (as in Lemma \ref{function_R}.(c)), the function $R_f$ is very regular outside the point $0$ even
if $f$ is not continuous.

\section{Integral formula}\label{Integral}
\setcounter{equation}{0}

In this section we prove our main result on the relation between the  evolution of
the localisation operators $f(\Phi/r)$ and the time operator $T_f$ defined below. We
begin with a technical lemma that will be used subsequently. Since this result could
also be useful in other situations, we present here a general version of it. The
symbol $\F$ stands for the Fourier transformation, and the measure $\underline\d x$
on $\R^n$ is chosen so that $\F$ extends to a unitary operator in $\ltwo(\R^n)$.

\begin{Proposition}\label{LemmaRafa}
Let $C\equiv(C_1,\ldots,C_n)$ and $D\equiv(D_1,\ldots,D_d)$ be two families of
mutually commuting self-adjoint operators in a Hilbert space $\Hrond$. Let $k\geq1$
be an integer, and assume that each $C_j$ is of class $C^k(D)$. Let $f\in\linf(\R^n)$,
set $g(x):=f(x)\<x_1\>^{2k}\cdots\<x_n\>^{2k}$, and suppose that the functions $g$
and
$$
x\mapsto(\F g)(x)\<x_1\>^{k+1}\cdots\<x_n\>^{k+1}
$$
are in $\lone(\R^n)$. Then the operator $f(C)$ belongs to $C^k(D)$. In particular, if
$f \in \S(\R^n)$ then $f(C)$ belongs to $C^k(D)$.
\end{Proposition}

\begin{proof}
For each $y\in\R^d$, we set $D_y:=\frac1{i|y|}(\e^{iy\cdot D}-1)$. Then we
know
from \cite[Lemma~6.2.3.(a)]{ABG} that it is sufficient to prove that
$\big\|\ad^k_{D_y}\big(f(C)\big)\big\|$ is bounded by a constant
independent of $y$.
By using the linearity of $\ad^k_{D_y}(\cdot)$ and \cite[Eq.~5.1.16]{ABG},
we get
\begin{align*}
&\ad^k_{D_y}\big(f(C)\big)\\
&=\ad^k_{D_y}\big(g(C)\<C_1\>^{-2k}\cdots\<C_n\>^{-2k}\big)\\
&=\int_{\R^n}\underline\d x\,(\F g)(x)
\ad^k_{D_y}\big(\e^{ix_1C_1}\<C_1\>^{-2k}\cdots\e^{ix_2C_n}\<C_n\>^{-2k}\big)\\
&=\sum_{k_1+\cdots+k_n=k}{\textsc c}_{k_1\cdots
k_n}\int_{\R^n}\underline\d x\,
(\F g)(x)\ad^{k_1}_{D_y}\big(\e^{ix_1C_1}\<C_1\>^{-2k}\big)\cdots
\ad^{k_n}_{D_y}\big(\e^{ix_2C_n}\<C_n\>^{-2k}\big),
\end{align*}
where ${\textsc c}_{k_1\cdots k_n}>0$ is some explicit constant. Furthermore,
since $C_j$ is of class $C^k(D)$, we know from
\cite[Eq.~6.2.13]{ABG} that
$$
\big\|\ad^{k_j}_{G_y}\big(\e^{ix_jC_j}\<C_j\>^{-2k}\big)\big\|
\le{\textsc c}_{k_j}\<x_j\>^{k+1},
$$
where ${\textsc c}_{k_j}\ge0$ is independent of $y$ and $x_j$. This
implies that
$$
\big\|\ad^k_{D_y}\big(f(C)\big)\big\|
\le{\rm Const.}\int_{\R^n}\underline\d x\,|(\F g)(x)|\<x_1\>^{k+1}\cdots\<x_n\>^{k+1}
\le{\rm Const.}\,,
$$
and the claim is proved.
\end{proof}

In Lemma \ref{Heigen}.(a) we have shown that the set $\kappa(H)$ is closed. So
we can define for each $t\ge0$ the set
$$
\D_t:=\big\{\v\in\dom(\<\Phi\>^t)\mid\v=\eta(H)\v
\textrm{ for some }\eta\in C^\infty_{\rm c}\big(\R\setminus\kappa(H)\big)\big\}.
$$
The set $\D_t$ is included in the subspace $\H_{\rm ac}(H)$ of absolute continuity
of $H$, due to Theorem \ref{not_bad}, and $\D_{t_1}\subset\D_{t_2}$ if $t_1\ge t_2$.
We refer the reader to Section \ref{Interpretation} for an account on density properties
of the sets $\D_t$.

In the sequel we consider the set of operators $\big\{H''_{jk}\big\}$ as the
components of a $d$-dimensional (Hessian) matrix which we denote by $H''$. Furthermore
we shall sometimes write $C^{-1}$ for an operator $C$ a priori not invertible.
In such a case, the operator $C^{-1}$ will always be restricted to a
set where it is well-defined. Namely, if $\dom$ is a set on which $C$ is invertible,
then we shall simply write ``$C^{-1}$ acting on $\dom$'' instead of using the
notation $C^{-1}|_{\dom}$.

\begin{Proposition}\label{lemma_T_f}
Let $H$ and $\Phi$ satisfy Assumptions \ref{chirimoya} and \ref{commute}. Let $f$
satisfy Assumption \ref{assumption_f} and assume that $R_f$ belongs to
$C^1(\R^d\setminus\{0\})$. Then the map
$$
t_f:\D_1\to\C,\quad\v\mapsto
t_f(\v):=-\12 \sum_{j}\big\{\big\langle\Phi_j\v,(\partial_jR_f)(H')\v\big\rangle
+\big\langle\big(\partial_jR_{\overline f}\big)(H')\v,\Phi_j\v\big\rangle\big\} ,
$$
is well-defined. Moreover, if $(\partial_jR_f)(H')\v$ belongs to $\dom(\Phi_j)$ for
each $j$, then the linear operator $T_f:\D_1\to\H$ defined by
\begin{equation}\label{nemenveutpas}
\textstyle
T_f\v:=-\12\Big(\Phi\cdot R_f'(H')+R_f'\big(\frac{H'}{|H'|})\cdot\Phi\;\!|H'|^{-1}
+iR_f'\big(\frac{H'}{|H'|}\big)\cdot(H''H')\;\!|H'|^{-3}\Big)\v
\end{equation}
satisfies $t_f(\v)=\langle v,T_f\v\rangle$ for each $\v\in\D_1$. In particular,
$T_f$ is a symmetric operator if $f$ is real and if $\D_1$ is dense in $\H$.
\end{Proposition}

\begin{Remark}\label{minusculeremarque}
{\rm Formula \eqref{nemenveutpas} is a priori rather complicated and one could be
tempted to replace it by the simpler formula $-\12\big(\Phi\cdot R_f'(H') + R_f'(H')\cdot\Phi\big)$. Unfortunately, a precise meaning of this expression is not available in general,
and its full derivation can only be justified in concrete examples.}
\end{Remark}

\begin{Remark}\label{petiteremarque}
{\rm If $\v\in\D_1$ and if $f$ either belongs to $\S(\R^d)$ or is radial, then the
assumption $(\partial_jR_f)(H')\v\in\dom(\Phi_j)$ holds for each $j$. Indeed, by
Lemma \ref{Heigen}.(d) there exists $\eta\in C^\infty_{\rm c}\big((0,\infty)\big)$
such that $(\partial_jR_f)(H')\v=(\partial_jR_f)(H')\eta\big((H')^2\big)\v$. By
Lemma \ref{function_R} and Proposition \ref{LemmaRafa}, it then follows that $(\partial_jR_f)(H')\eta\big((H')^2\big)\in C^1(\Phi_j)$, which implies the
statement.}
\end{Remark}

\begin{proof}[Proof of Proposition \ref{lemma_T_f}]
Let $\v\in\D_1$. Then Lemma \ref{Heigen}.(d) implies that there exists a function
$\eta\in C^\infty_{\rm c}\big((0,\infty)\big)$ such that
$$
(\partial_jR_f)(H')\v=(\partial_jR_f)(H')\eta\big((H')^2\big)\v.
$$
Thus $\|(\partial_jR_f)(H')\v\|\le{\rm Const.}\,\|\v\|$, and we have
$$
|t_f(\v)|\le{\rm Const.}\,\|\v\|\cdot\|\langle\Phi\rangle\v\|,
$$
which implies the first part of the claim.

For the second part of the claim, it is sufficient to show that
\begin{align*}
&\sum_j\big\langle\big(\partial_jR_{\overline f}\big)(H')\v,\Phi_j\v\big\rangle\\
&=\textstyle\big\langle\v,\big\{R_f'\big(\frac{H'}{|H'|})\cdot\Phi\;\!|H'|^{-1}
+iR_f'\big(\frac{H'}{|H'|}\big)\cdot(H''H')\;\!|H'|^{-3}\big\}\v\big\rangle.
\end{align*}
Using Formula \eqref{R_f_alpha}, Lemma \ref{Heigen}.(d), and
\cite[Eq.~4.3.2]{Dav95}, one gets
\begin{align*}
&\sum_j\big\langle\big(\partial_jR_{\overline f}\big)(H')\v,\Phi_j\v\big\rangle\\
&=\sum_j\big\langle\textstyle(\partial_jR_{\overline f})\big(\frac{H'}{|H'|}\big)
|H'|^{-1}\v,\Phi_j\v\big\rangle\\
&=\sum_j\lim_{\varepsilon\searrow0}
\big\langle\textstyle\big(\partial_jR_{\overline f}\big)\big(\frac{H'}{|H'|}\big)\v,
[(H')^2+\varepsilon]^{-1/2}\Phi_j\v\big\rangle\\
&=\textstyle\big\langle\v,
R_f'\big(\frac{H'}{|H'|})\cdot\Phi\;\!|H'|^{-1}\v\big\rangle\\
&\qquad+\pi^{-1}\sum_j\lim_{\varepsilon\searrow0}\int_0^\infty\d t\,t^{-1/2}
\big\langle\textstyle\big(\partial_jR_{\overline f}\big)\big(\frac{H'}{|H'|}\big)\v,
\big[[(H')^2+\varepsilon+t]^{-1},\Phi_j\big]\v\big\rangle.
\end{align*}
Now, by using Assumption \ref{chirimoya} and Lemma \ref{undos} one obtains that
$$
\big[[(H')^2+\varepsilon+t]^{-1},\Phi_j\big]\v
=2i\big[(H')^2+\varepsilon+t\big]^{-2}(H''H')_j\;\!\v.
$$
It follows that
\begin{align*}
&\pi^{-1}\sum_j\lim_{\varepsilon\searrow0}\int_0^\infty\d t\,t^{-1/2}
\big\langle\textstyle\big(\partial_jR_{\overline f}\big)\big(\frac{H'}{|H'|}\big)\v,
2i[(H')^2+\varepsilon+t]^{-2}(H''H')_j\v\big\rangle\\
&=\sum_j\lim_{\varepsilon\searrow0}
\big\langle\textstyle\big(\partial_jR_{\overline f}\big)\big(\frac{H'}{|H'|}\big)\v,
i[(H')^2+\varepsilon]^{-3/2}(H''H')_j\v\big\rangle\\
&=\textstyle\big\langle\v,iR_f'\big(\frac{H'}{|H'|}\big)\cdot(H''H')
\;\!|H'|^{-3}\v\big\rangle,
\end{align*}
and thus
$$
\sum_j\big\langle\big(\partial_jR_{\overline f}\big)(H')\v,\Phi_j\v\big\rangle
=\textstyle\big\langle\v,\big\{R_f'\big(\frac{H'}{|H'|})\cdot\Phi\;\!|H'|^{-1}
+iR_f'\big(\frac{H'}{|H'|}\big)\cdot(H''H')\;\!|H'|^{-3}\big\}\v\big\rangle.
$$
\end{proof}

Suppose for a while that $f$ is radial. Then one has $(\partial_jR_f)(x)=-x^{-2}x_j$
due to Lemma \ref{function_R}.(c), and Formula \eqref{nemenveutpas} holds by Remark \ref{petiteremarque}. This implies that $T_f$ is equal to
\begin{equation}\label{T}
\textstyle
T:=\12 \Big(\Phi\cdot\frac{H'}{(H')^2}+\frac{H'}{|H'|}\cdot\Phi\;\!|H'|^{-1}
+\frac{iH'}{(H')^4}\cdot(H''H')\Big)
\end{equation}
on $\D_1$.

The next theorem is our main result; it relates the evolution of localisation operators
$f(\Phi/r)$ to the operator $T_f$. In its proof, we freely use the notations of
\cite{ABG} for some regularity classes with respect to the unitary group generated by
$\Phi$. For us, a function $f:\R^d\to\C$ is even if $f(x)=f(-x)$ for \aew $x\in\R^d$.

\begin{Theorem}\label{for_Schwartz}
Let $H$ and $\Phi$ satisfy Assumptions \ref{chirimoya} and \ref{commute}. Let
$f\in\S(\R^d)$ be an even function such that $f=1$ on a neighbourhood of $0$.
Then we have for each $\v\in\D_2$
\begin{equation}\label{bisoufe}
\lim_{r\to\infty}\12 \int_0^\infty\d t\,\big\langle\v,
\big[\e^{-itH}f(\Phi/r)\e^{itH}-\e^{itH}f(\Phi/r)\e^{-itH}
\big]\v\big\rangle\\
=t_f(\v).
\end{equation}
\end{Theorem}

Note that the integral on the l.h.s. of \eqref{bisoufe} is finite for each $r>0$
since $f(\Phi/r)$ can be factorized as
$$
f(\Phi/r)\equiv|f(\Phi/r)|^{1/2}\cdot\sgn[f(\Phi/r)]\cdot|f(\Phi/r)|^{1/2},
$$
with $|f(\Phi/r)|^{1/2}$ locally $H$-smooth on $\R\setminus\kappa(H)$ by Theorem
\ref{not_bad}. Furthermore, since Remark \ref{petiteremarque} applies, the r.h.s.
can also be written as the expectation value $\langle\v,T_f\v\rangle$.

\begin{proof}
(i) Let $\v\in\D_2$, take a real
$\eta\in C^\infty_{\rm c}\big(\R\setminus\kappa(H)\big)$ such that
$\eta(H)\v=\v$, and set $\eta_t(H):=\e^{itH}\eta(H)$. Then we have
\begin{align}
&\big\langle\v,\big[\e^{itH}f(\Phi/r)\e^{-itH}
-\e^{-itH}f(\Phi/r)\e^{itH}\big]\v\big\rangle\nonumber\\
&=\int_{\R^d}\underline\d x\,(\F f)(x)
\big\langle\v,\big[\eta_t(H)\e^{i\frac xr\cdot\Phi}\eta_{-t}(H)
-\eta_{-t}(H)\e^{i\frac xr\cdot\Phi}\eta_t(H)\big]\v\big\rangle\nonumber\\
&=\int_{\R^d}\underline\d x\,(\F f)(x)\big\langle\v,
\big[\e^{i\frac xr\cdot\Phi}\eta_t\big(H(\textstyle\frac xr)\big)\eta_{-t}(H)
-\eta_{-t}(H)\eta_t\big(H(-\textstyle\frac xr)\big)\e^{i\frac xr\cdot\Phi}\big]
\v\big\rangle\nonumber\\
&=\int_{\R^d}\underline\d x\,(\F f)(x)\big\langle\v,
\big\{\big(\e^{i\frac xr\cdot\Phi}-1\big)\eta_t\big(H(\textstyle\frac xr)\big)
\eta_{-t}(H)\label{pave}\\
&\qquad\qquad+\eta_{-t}(H)\big[\eta_t\big(H(\textstyle\frac xr)\big)
-\eta_t\big(H(-\textstyle\frac xr)\big)\big]
-\eta_{-t}(H)\eta_t\big(H(-\textstyle\frac xr)\big)
\big(\e^{i\frac xr\cdot\Phi}-1\big)\big\}\v\big\rangle.\nonumber
\end{align}
Since $f$ is even, $\F f$ is also even, and
$$
\int_{\R^d}\underline\d x\,(\F f)(x)\<\v,
\eta_{-t}(H)\big[\eta_t\big(H(\textstyle\frac xr)\big)
-\eta_t\big(H(-\textstyle\frac xr)\big)\big]\v\>=0.
$$
Thus Formula \eqref{pave}, Lemma \ref{undos}, and the change of
variables $\mu:=t/r$, $\nu:=1/r$, give
\begin{equation}\label{limit}
\lim_{r\to\infty}\12 \int_0^\infty\d t\,\big\langle\v,
\big[\e^{-itH}f(\Phi/r)\e^{itH}-\e^{itH}f(\Phi/r)\e^{-itH}\big]\v\big\rangle
=-\12 \lim_{\nu\searrow0}\int_0^\infty\d\mu\int_{\R^d}\underline\d x\,
K(\nu,\mu,x),
\end{equation}
where
\begin{align*}
K(\nu,\mu,x):=(\F f)(x)\big\langle\v,
\big\{&\textstyle\frac1\nu\big(\e^{i\nu x\cdot\Phi}-1\big)
\eta(H(\nu x))\e^{i\frac\mu\nu[H(\nu x)-H]}\nonumber\\
&-\eta(H(-\nu x))\e^{i\frac\mu\nu[H(-\nu x)-H]}
\textstyle\frac1\nu\big(\e^{i\nu x\cdot\Phi}-1\big)\big\}\v\big\rangle.
\end{align*}

(ii) To prove the statement, we shall show that one may interchange the limit and the integrals in \eqref{limit}, by invoking Lebesgue's dominated convergence theorem. This
will be done in (iii) below. Here we pursue the calculations assuming that these interchanges are justified.

We know from Assumption \ref{chirimoya} that $H$ is of class $C^2(\Phi_j)$ (and thus
of class $C^{1,1}(\Phi_j)$) for each $j\in\{1,\ldots,d\}$. Since the domain
of $H$ is invariant under the group generated by $\Phi_j$, it follows then from \cite[Thm.~6.3.4.(b)]{ABG} that $H$ belongs to $C^{1,1}(\Phi_j,\G,\G^*)$, where
$\G$ denotes the space $\dom(H)$ endowed with the graph topology. In particular,
$H$ belongs to $C^1_u(\Phi_j,\G,\G^*)$; namely, the map
$\R\ni\nu\mapsto H(\nu e_j)\in\B(\G,\G^*)$ is continuously differentiable in the
uniform topology. Therefore the map
$$
\textstyle
\R\setminus\{0\}\ni\nu\mapsto\frac1\nu[H(\nu e_j)-H]\in\B(\G,\G^*)
$$
extends to a continuous map defined on $\R$ and taking value $H'_j$ at $\nu=0$.

Now, the exponential map $B\mapsto\e^{iB}$ is continuous from $\B(\G,\G^*)$ to $\B(\G,\G^*)$. So, the composed map
$$
\R\ni\nu\mapsto\e^{\frac i\nu[H(\nu e_j)-H]} \in \B(\G,\G^*)
$$
is also continuous, and takes value $\e^{iH'_j}$ at $\nu =0$. By linearity and
by taking Lemma \ref{undos} into account, one
finally obtains in $\B(\G,\G^*)$
$$
\lim_{\nu\searrow0}\e^{i\frac\mu\nu[H(\nu x)-H]}=\e^{i\mu x\cdot H'}.
$$
It follows that for any $\v,\psi\in\G$, one has
$$
\lim_{\nu\searrow0}\big\langle\psi,\e^{i\frac\mu\nu[H(\nu x)-H]}\v\big\rangle
=\big\langle\psi,\e^{i\mu x\cdot H'}\v\big\rangle.
$$
In fact, since the operators $H,H(\nu x)$ and $H'_j$ are self-adjoint this
equality even holds for $\v,\psi\in\H$, but we do not need such an extension.
This identity, together with the symmetry of $f$, Lemma \ref{function_R}.(a), and
Proposition \ref{lemma_T_f}, implies that for $\v\in\D_2$
\begin{align*}
&\lim_{r\to\infty} \12 \int_0^\infty\d t\,\big\langle\v,\big[\e^{-itH}f(\Phi/r)\e^{itH}- \e^{itH}f(\Phi/r)\e^{-itH}\big]\v\big\rangle\\
&={\textstyle -\frac{i}{2}}\int_0^\infty\d\mu\,\int_{\R^d}\underline\d x\,(\F f)(x)
\big\{\big\langle\(x\cdot\Phi\)\v,\e^{i\mu x\cdot H'}\v\big\rangle
-\big\langle\v,\e^{-i\mu x\cdot H'}\(x\cdot\Phi\)\v\big\rangle\big\}\\
&=-\12 \sum_{j}\int_0^\infty\d\mu\,\int_{\R^d}\underline\d x\,[\F(\partial_jf)](x)
\big[\big\langle\Phi_j\v,\e^{i\mu x\cdot H'}\v\big\rangle
+\big\langle\v,\e^{i\mu x\cdot H'}\Phi_j\v\big\rangle\big]\\
&=-\12 \sum_{j}\int_0^\infty\d\mu\,
\big[\big\langle\Phi_j\v,(\partial_jf)\big(\mu H'\big)\v\big\rangle
+\big\langle\big(\partial_j\overline f\big)
\big(\mu H'\big)\v,\Phi_j\v\big\rangle\big]\\
&=t_f(\v).
\end{align*}

(iii) To interchange the limit $\nu\searrow0$ and the integration over $\mu$ in
\eqref{limit}, one has to bound $\int_{\R^d}\underline\d x\,K(\nu,\mu,x)$ uniformly
in $\nu$ by a function in $\lone\big((0,\infty),\d\mu\big)$. We begin with the first
term of $\int_{\R^d}\underline\d x\,K(\nu,\mu,x)$:
$$
K_1(\nu,\mu):=\int_{\R^d}\underline\d x\,(\F f)(x)
\<\langle\Phi\rangle^2\v,\textstyle\frac1\nu\big(\e^{i\nu x\cdot\Phi}-1\big)
\langle\Phi\rangle^{-2}\eta(H(\nu x))
\e^{i\frac\mu\nu[H(\nu x)-H]}\v\>.
$$
Observe that for each multi-index $\alpha\in\N^d$ with $|\alpha|\leq 2$ one has
\begin{equation}\label{nu_and_x}
\big\|\textstyle\partial^\alpha_x\frac1\nu
\big(\e^{i\nu x\cdot\Phi}-1\big)\langle\Phi\rangle^{-2}\big\|
\le{\rm Const.}\;\!\langle x\rangle,
\end{equation}
where the derivatives are taken in the strong topology and where the constant is
independent of $\nu\in(-1,1)$. Since $\F f\in\S(\R^d)$ it follows that
\begin{equation}\label{mu_small}
\big|K_1(\nu,\mu)\big|\le{\rm Const.},
\end{equation}
and thus $K_1(\nu,\mu)$ is bounded uniformly in $\nu$ by a function in
$\lone\big((0,1],\d\mu\big)$.

For the case $\mu>1$ we first remark that there exists a compact set
$J\subset\R\setminus \kappa(H)$ such that $\v=E^H(J)\v$. There
also exists $\zeta\in C^\infty_{\rm c}\big((0,\infty)\big)$ such that
$\zeta\big((H')^2\big)\eta(H)=\eta(H)$ due to Lemma \ref{Heigen}.(d).
It then follows that
$$
\eta(H(\nu x))\e^{i\frac\mu\nu[H(\nu x)-H]}\v
=\zeta\big(H'(\nu x)^2\big)\eta(H(\nu x))\e^{i\frac\mu\nu[H(\nu x)-H]}\v.
$$
Moreover, from Assumption \ref{commute}, we also get that
$$
B_{\nu,\mu}^J(x)\v
:=E^H(J)\e^{i\frac\mu\nu[H(\nu x)-H]}E^H(J)\v
=\e^{i\frac\mu\nu[H(\nu x)-H]}\v.
$$
So, $K_1(\nu,\mu)$ can be rewritten as
$$
\int_{\R^d}\underline\d x\,(\F f)(x)
\<\langle\Phi\rangle^2\v,\textstyle\frac1\nu\big(\e^{i\nu x\cdot\Phi}-1\big)
\langle\Phi\rangle^{-2}\zeta\big(H'(\nu x)^2\big)\eta(H(\nu x))B_{\nu,\mu}^J(x)\v\>.
$$
Now, it is easily shown by using Assumption \ref{chirimoya} and Lemma \ref{undos}
that the function $B_{\nu,\mu}^J:\R^d\to\B(\H)$ is differentiable with derivative
equal to
$$
\big(\partial_jB_{\nu,\mu}^J\big)(x)=i\mu H'_j(\nu x)B_{\nu,\mu}^J(x).
$$
Furthermore, the bounded operator
$$
\textstyle
A_{j,\nu}(x):=(\F f)(x)\frac1\nu\big(\e^{i\nu x\cdot\Phi}-1\big)
\langle\Phi\rangle^{-2}H'_j(\nu x)|H'(\nu x)|^{-2}\zeta\big(H'(\nu x)^2\big)
\eta(H(\nu x))
$$
satisfies for each integer $k\ge1$ the bound
$$
\big\|A_{j,\nu}(x)\big\|\le{\rm Const.}\;\!\langle x\rangle^{-k},
$$
due to Assumption \ref{chirimoya}, Lemma \ref{undos}, Equation \eqref{nu_and_x}
and the rapid decay of $\F f$. Thus $K_1(\nu,\mu)$ can be written as
$$
K_1(\nu,\mu)=-i\mu^{-1}\sum_{j}\int_{\R^d}\underline\d x
\<\langle\Phi\rangle^2\v,A_{j,\nu}(x)\big(\partial_jB_{\nu,\mu}^J\big)(x)\v\>.
$$
Moreover, direct calculations using Equation \eqref{nu_and_x} and Proposition
\ref{LemmaRafa} show that the map  $\R^d\ni x\mapsto A_{j,\nu}(x)\in \B(\H)$ is
twice strongly differentiable and satisfies
$$
\big\|(\partial_jA_{j,\nu})(x)\big\|\le{\rm Const.}\;\!\langle x\rangle^{-k}
$$
and
\begin{equation}\label{chmolo}
\big\|\partial_\ell\big\{(\partial_jA_{j,\nu})
H'_\ell(\nu\;\!\cdot\;\!)(H'(\nu\;\!\cdot\;\!))^{-2}\big\}(x)\big\|
\le{\rm Const.}\(1+|\nu|\)\langle x\rangle^{-k}
\end{equation}
for any integer $k\ge1$. Therefore one can perform two successive integrations by parts (with vanishing boundary contributions) and obtain
\begin{align*}
K_1(\nu,\mu)&=i\mu^{-1}\sum_{j}\int_{\R^d}\underline\d x\<\langle\Phi\rangle^2\v
,(\partial_jA_{j,\nu})(x)B_{\nu,\mu}^J(x)\v\>\\
&=-\mu^{-2}\sum_{j,\ell}\int_{\R^d}\underline\d x\,\big\langle\langle\Phi\rangle^2\v
,\partial_\ell\big\{(\partial_jA_{j,\nu})
H'_\ell(\nu\;\!\cdot\;\!)(H'(\nu\;\!\cdot\;\!))^{-2}\big\}(x)
B_{\nu,\mu}^J(x)\v\big\rangle.\label{two_by_parts}
\end{align*}
This together with Formula \eqref{chmolo} implies for each $\nu<1$ and each $\mu>1$
that
\begin{equation}\label{mu_big}
\big|K_1(\nu,\mu)\big|\le{\rm Const.}\,\mu^{-2}.
\end{equation}
The combination of the bounds \eqref{mu_small} and \eqref{mu_big} shows that
$K_1(\nu,\mu)$ is bounded uniformly for $\nu<1$ by a function in
$\lone\big((0,\infty),\d\mu\big)$. Since similar arguments shows that the same holds for the
second term of $\int_{\R^d}\underline\d x\,K(\nu,\mu,x)$, one can interchange the limit
$\nu\searrow0$ and the integration over $\mu$ in \eqref{limit}.

The interchange of the limit $\nu\searrow0$ and the integration over $x$
in \eqref{limit} is justified by the bound
$$
\big|K(\nu,\mu,x)\big|\le{\rm Const.}\,\big|x(\F f)(x)\big|,
$$
which follows from Formula \eqref{nu_and_x}.
\end{proof}

When the localisation function $f$ is radial, the operator $T_f$ is equal to the
operator $T$, which is independent of $f$. The next result, which depicts this
situation of particular interest, is a direct consequence of
Lemma \ref{function_R}.(c) and Theorem \ref{for_Schwartz}

\begin{Corollary}\label{radial_case}
Let $H$ and $\Phi$ satisfy Assumptions \ref{chirimoya} and \ref{commute}. Let
$f\in\S(\R^d)$ be a radial function such that $f=1$ on a neighbourhood of $0$. Then
we have for each $\v\in\D_2$
\begin{equation}\label{energy}
\lim_{r\to\infty}\12 \int_0^\infty\d t\,\big\langle\v,\big[\e^{-itH}f(\Phi/r)\e^{itH}-\e^{itH}f(\Phi/r)\e^{-itH}
\big]\v\big\rangle\\
=\langle\v,T\v\rangle,
\end{equation}
with $T$ defined by \eqref{T}.
\end{Corollary}

\section{Interpretation of the integral formula}\label{Interpretation}
\setcounter{equation}{0}

This section is devoted to the interpretation of Formula \eqref{bisoufe} and to the
description of the sets $\D_t$. We begin by stressing some properties of the subspace $\K:=\ker\big((H')^2\big)$ of $\H$, which plays an important role in the sequel.

\begin{Lemma}\label{surK}
\begin{enumerate}
\item[(a)] The eigenvectors of $H$ belong to $\K$,
\item[(b)] If $\v\in\K$, then the spectral support of $\v$ with respect to $H$ is
contained in $\kappa(H)$,
\item[(c)] For each $t\geq 0$, the set $\K$ is orthogonal to $\D_t$,
\item[(d)] For each $t\geq 0$, the set $\D_t$ is dense in $\H$ only if $\K$ is trivial.
\end{enumerate}
\end{Lemma}

\begin{proof}
As observed in the proof of Lemma \ref{Heigen}, if $\lambda$ is an eigenvalue
of $H$ then one has $E^H(\{\lambda\})H'_jE^H(\{\lambda\})=0$ for each $j$. If
$\v_\lambda$ is some corresponding eigenvector, it follows that $H'_j\v_\lambda=E^H(\{\lambda\})H'_jE^H(\{\lambda\})\v_\lambda=0$. Thus, all
eigenvectors of $H$ belong to the kernel of $H'_j$, and a fortiori to the kernels of $(H'_j)^2$ and $(H')^2$.

Now, let $\v\in\K$ and let $J$ be the minimal closed subset of $\R$ such that
$E^H(J)\v=\v$. It follows then from Definition \ref{surkappa} that $J\subset\kappa(H)$.
This implies that $\v\bot\D_t$, and thus $\K\bot\D_t$. The last statement is a straightforward consequence of point (c).
\end{proof}

Let us now proceed to the interpretation of Formula \eqref{bisoufe}. We consider
first the term $t_f(\v)$ on the r.h.s., and recall that $f$ is an even element of
$\S(\R^d)$ with $f=1$ in a neighbourhood of $0$. We also assume that $f$ is real.

Due to Remark \ref{petiteremarque} with $\v\in\D_1$, the term $t_f(\v)$ reduces
to the expectation value $\langle\v,T_f\v\rangle$, with $T_f$ given by
\eqref{nemenveutpas}. Now, a direct calculation using Formulas \eqref{minusone},
\eqref{R_f_alpha}, and \eqref{nemenveutpas} shows that the operators $T_f$ and $H$
satisfy in the form sense on $\D_1$ the canonical commutation relation
\begin{equation}\label{T_H}
\textstyle
\big[T_f,H\big]=i.
\end{equation}
Therefore, since the group $\{\e^{-itH}\}_{t\in\R}$ leaves $\D_1$ invariant, the
following equalities hold in the form sense on $\D_1$:
$$
T_f\e^{-itH}
=\e^{-itH}T_f+\big[T_f,\e^{-itH}\big]
=\e^{-itH} T_f
-i\int_0^t\d s\,\e^{-i(t-s)H}\big[T_f,H\big]\e^{-isH}
=\textstyle\e^{-itH}\big(T_f+t\big).
$$
In other terms, one has
\begin{equation}\label{weakWeyl}
\textstyle\big\langle\psi,T_f\e^{-itH}\v\big\rangle
=\big\langle\psi,\e^{-itH}\big(T_f+t\big)\v\big\rangle
\end{equation}
for each $\psi,\v\in\D_1$, and the operator $T_f$ satisfies on $\D_1$ the
so-called infinitesimal Weyl relation in the weak sense \cite[Sec.~3]{JM80}. Note
that we have not supposed that $\D_1$ is dense. However, if $\D_1$ is dense in
$\H$, then the infinitesimal Weyl relation in the strong sense holds:
\begin{equation}\label{strongWeyl}
\textstyle
T_f\e^{-itH}\v=\e^{-itH}\big(T_f+t\big)\v,\qquad\v\in\D_1.
\end{equation}
This relation, also known as $T_f\,$-weak Weyl relation \cite[Def.~1.1]{Miy01},
has deep implications on the spectral nature of $H$ and on the form of $T_f$ in
the spectral representation of $H$. Formally, it suggests that
$T_f=i\frac\d{\d H}$, and thus $-iT_f$ can be seen as the operator of
differentiation with respect to the Hamiltonian $H$. Moreover, being a weak
version of the usual Weyl relation, Relation \eqref{strongWeyl} also suggests
that the spectrum of $H$ may not differ too much from a purely absolutely
continuous spectrum. These properties are now discussed more rigorously in
particular situations. In the first two cases, the density of $\D_1$ in $\H$ is
assumed, and so the point spectrum of $H$ is empty by Lemma \ref{surK}.

{\bf Case 1 ($\boldsymbol{T_f}$ essentially self-adjoint):} If the set $\D_1$ is
dense in $\H$, and $T_f$ is essentially self-adjoint on $\D_1$, then it has been
shown in \cite[Lemma~4]{JM80} that \eqref{strongWeyl} implies that the pair $\{\overline{T_f},H\}$ satisfies the usual Weyl relation, \ie
$$
\e^{isH}\e^{it\overline{T_f}}
=\e^{ist}\e^{it\overline{T_f}}\e^{isH},\qquad s,t\in\R.
$$
It follows by the Stone-von Neumann theorem \cite[VIII.14]{RSI} that there exists
a unitary operator $\U:\H\to\ltwo(\R;\C^N,\d\lambda)$, with $N$ finite or infinite,
such that $\U\e^{it\overline{T_f}}\U^*$ is the operator of translation by $t$,
and $\U\e^{is H}\U^*$ is the operator of multiplication by $\e^{is\lambda}$. In
terms of the generator $H$, this means that $\U H\U^*=\lambda$, where ``$\lambda$''
stands for the multiplication operator by $\lambda$ in $\ltwo(\R;\C^N,\d\lambda)$.
Therefore the spectrum of $H$ is purely absolutely continuous and covers the whole
real line. Moreover, we have for each $\psi\in\H$ and $\v\in\D_1$
$$
\langle\psi,T_f\v\rangle
=\langle\psi,\overline{T_f}\v\rangle
=\int_\R\d\lambda\,\textstyle\big\langle(\U\psi)(\lambda),
i\,\frac{\d(\U\v)}{\d\lambda}(\lambda)\big\rangle_{\C^N},
$$
where $\frac\d{\d\lambda}$ denotes the distributional derivative (see for instance
\cite[Rem.~1]{AC87} for an interpretation of the derivative $\frac\d{\d\lambda}$).

{\bf Case 2 ($\boldsymbol{T_f}$ symmetric):} If the set $\D_1$ is dense in $\H$,
then we know from Proposition \ref{lemma_T_f} and Remark \ref{petiteremarque} that $T_f$
is symmetric. In such a situation, \eqref{strongWeyl} once more implies that the
spectrum of $H$ is purely absolutely continuous \cite[Thm.~4.4]{Miy01}, but it may
not cover the whole real line. We expect that the operator $T_f$ is still equal to
$i\frac\d{\d\lambda}$ (on a suitable subspace) in the spectral representation of
$H$, but we have not been able to prove it in this generality. However, this
property holds in most of the examples presented below. If $T_f$ and $H$ satisfy
more assumptions, then more can be said (see for instance \cite{Sch83}).

{\bf Case 3 ($\boldsymbol{T_f}$ not densely defined):} If $\D_1$ is not dense in
$\H$, then we are not aware of general works using a relation like
\eqref{weakWeyl} to deduce results on the spectral nature of $H$ or on the form
of $T_f$ in the spectral representation of $H$. In such a case, we only know from
Theorem \ref{not_bad} that the spectrum of $H$ is purely absolutely continuous in
$\sigma(H)\setminus\kappa(H)$, but we have no general information on the form of
$T_f$ in the spectral representation of $H$. However, with a suitable additional
assumption the analysis can be continued. Indeed, consider the orthogonal
decomposition $\H:=\K\oplus\G$, with $\K\equiv\ker \big((H')^2\big)$. Then the
operators $H$, $H'_j$, and $H''_{k\ell}$ are all reduced by this decomposition,
due to the commutation assumption \ref{commute}. If we assume additionally that
$T_f\D_1\subset\G$, then the analysis can be performed in the subspace $\G$.

Since $\D_1\subset\G$ by Lemma \ref{surK}, the additional hypothesis allows us to
consider the restriction of $T_f$ to $\G$, which we denote by $\TT_{\!f}$. Let
also $\HH$, $\HH'_j$, and $\HH''_{k\ell}$ denote the restrictions of the
corresponding operators in $\G$. We then set
$$
\DD_t:=\big\{\v\in\dom(\langle\Phi\rangle^t)\cap \G\mid\v=\eta(\HH)\v
\textrm{ for some }\eta\in C^\infty_{\rm c}\big(\R\setminus\kappa(H)\big)\big\}
\subset\G,
$$
and observe that the equality \eqref{T_H} holds in the form sense on $\DD_1$. In
other words, \eqref{T_H} can be considered in the reduced Hilbert space $\G$
instead of $\H$. The interest of the above decomposition comes from the following
fact: If $\DD_1$ is dense in $\G$ (which is certainly more likely than in $\H$),
then $\TT_{\!f}$ is symmetric and the situation reduces to the case $2$ with the
operators $\HH$ and $\TT_{\!f}$. If in addition $\TT_{\!f}$ is essentially
self-adjoint on $\DD_1$, the situation even reduces to the case $1$ with the
operators $\HH$ and $\TT_{\!f}$. In both situations, the spectrum of $\HH$ is
purely absolutely continuous. In Section \ref{SecEx}, we shall present $2$
examples corresponding to these situations.

\begin{Remark}\label{queneni}
{\rm The implicit condition $T_f \D_1\subset\G$ can be made more explicit. For example,
if the collection $\Phi$ is reduced by the decomposition $\H=\K\oplus\G$, then the
condition holds (and \eqref{bisoufe} also holds on $\DD_2$). More generally, if
$\Phi_j\D_1\subset\G$ for each $j$, then the condition holds. Indeed, if $\v\in\D_1$
one knows from Remark \ref{petiteremarque} that
$(\partial_jR_f)(H')\v\in\dom(\langle\Phi\rangle)$, and one can prove similarly that $|H'|^{-1}\v \in\dom(\langle\Phi\rangle)$. Furthermore, there exists
$\eta\in C^\infty_{\rm c}\big(\R\setminus\kappa(H)\big)$ such that
$(\partial_jR_f)(H')\v=\eta(H)(\partial_jR_f)(H')\v$ and
$|H'|^{-1}\v=\eta(H)|H'|^{-1}\v$, which means that both vectors
$\partial_j R_f(H')\v$ and $|H'|^{-1}\v$ belong to $\D_1$. It follows that
$T_f\v\in\G$ by taking the explicit form \eqref{nemenveutpas} of $T_f$ into account.}
\end{Remark}

Let us now concentrate on the other term in Formula \eqref{bisoufe}.
If we consider the operators $\Phi_j$ as the components of an abstract
position operator $\Phi$, then the l.h.s. of Formula \eqref{bisoufe} has the
following meaning: For $r$ fixed, it can be interpreted as the difference of
times spent by the evolving state $\e^{-itH}\v$ in the past (first term)
and in the future (second term) within the region defined by the localisation
operator $f(\Phi/r)$. Thus, Formula \eqref{bisoufe} shows that this difference
of times tends as $r\to\infty$ to the expectation value in $\v$ of the operator
$T_f$.

On the other hand, let us consider a quantum scattering pair $\{H,H+V\}$, with
$V$ an appropriate perturbation of $H$. Let us also assume that the corresponding
scattering operator $S$ is unitary, and recall that $S$ commute with $H$. In this
framework, the global time delay $\tau(\v)$ for the state $\v$ defined in terms of
the localisation operators $f(\Phi/r)$ can usually be reexpressed as follows: it is
equal to the l.h.s. of \eqref{bisoufe} minus the same quantity with
$\v$ replaced by $S\v$. Therefore, if $\v$ and $S\v$ are elements of $\D_2$, then
the time delay for the scattering pair $\{H,H+V\}$ should satisfy the equation
\begin{equation}\label{Eisenbud}
\tau(\v)=-\langle\v,S^*[T_f,S]\v\rangle.
\end{equation}
In addition, if $T_f$ acts in the spectral representation of $H$ as a differential
operator $i\frac\d{\d H}$, then $\tau(\v)$ would verify, in our complete abstract
setting, the Eisenbud-Wigner formula
$$
\textstyle
\tau(\v)=\big\langle\v,-iS^*\frac{\d S}{\d H}\;\!\v\big\rangle.
$$

Summing up, as soon as the position operator $\Phi$ and the operator $H$ satisfy
Assumptions \ref{chirimoya} and \ref{commute}, then our study establishes a
preliminary relation between time operators $T_f$ given by \eqref{nemenveutpas}
and the theory of quantum time delay. Many concrete examples discussed in the
literature \cite{AC87,ACS87,AJ07,GT07,MSA92,Tie06,Tie09_3} turn out to fit in the
present framework, and several old or new examples are presented in the following
section. Further investigations in relation with the abstract Formula
\eqref{Eisenbud} will be considered elsewhere.

Now, most of the above discussion depends on the size of $\D_1$ in $\H$, and
implicitly on the size of $\kappa(H)$ in $\sigma(H)$. We collect some information
about these sets. It has been proved in Lemma \ref{Heigen}.(d) that $\kappa(H)$
is closed and corresponds to the complement in $\sigma(H)$ of the Mourre set (see
the comment after Definition \ref{surkappaA}). It always contains the eigenvalues
of $H$. Furthermore, since the spectrum of $H$ is absolutely continuous on $\sigma(H)\setminus\kappa(H)$, the support of the singularly continuous spectrum,
if any, is contained in $\kappa(H)$. In particular, if $\kappa(H)$ is discrete,
then $H$ has no singularly continuous spectrum. Thus, the determination of the
size of $\kappa(H)$ is an important issue for the spectral analysis of $H$. More
will be said in the concrete examples of the next section.

Let us now turn to the density properties of the sets $\D_t$. For this, we recall
that a subset $K\subset\R$ is said to be uniformly discrete if
$$
\inf\{|x-y|\mid x,y\in K\hbox{ and }x\neq y\}>0.
$$

\begin{Lemma}\label{density}
Assume that $\kappa(H)$ is uniformly discrete. Then
\begin{enumerate}
\item[(a)] $\D_0$ is dense in $\H_{\rm ac}(H)$,
\item[(b)] If $\sigma_{\rm p}(H)=\varnothing$ and if $H$
is of class $C^k(\Phi)$ for some integer $k$, then $\D_t$ is dense in $\H$ for any $t \in [0,k)$.
\end{enumerate}
\end{Lemma}

\begin{proof}
(a) Let $\v\in\H_{\rm ac}(H)$ and $\varepsilon>0$. Then there exists a finite
interval $[a,b]$ such that
$
\big\|\big[1-E^H([a,b])\big]\v\big\|\leq\varepsilon/2.
$
Since $\kappa(H)$ is uniformly discrete, the set $\kappa(H)\cap(a,b)$ contains only a
finite number $N$ of points $x_1<x_2 <\dots<x_N$. Let us set $x_0:=a$ and
$x_{N+1}:=b$. Since $\v\in\H_{\rm ac}$, there exists $\delta>0$ such that
$x_j+\delta <x_{j+1}-\delta$ for each $j\in\{0,\ldots,N\}$, and
$
\|E^H(L_\delta)\v\big\|\leq\varepsilon/2
$,
where
$$
L_\delta:=\{x \in [a,b] \mid |x-x_j|\leq\delta\hbox{ for each }j=0,1,\ldots,N+1\}.
$$
Now, for any $j\in\{0,\ldots,N\}$ there exist
$
\eta_j,\widetilde{\eta}_j\in C_{\rm c}^\infty\big((x_j,x_{j+1});[0,1]\big)
$
such that $\widetilde{\eta}_j(x)=1$ for $x\in[x_j+\delta,x_{j+1}-\delta]$ and
$\eta_j\widetilde{\eta}_j=\widetilde{\eta}_j$. Therefore, if $\eta:=\sum_{j=0}^N\eta_j$, $\widetilde{\eta}:=\sum_{j=0}^N\widetilde{\eta}_j$
and $\psi:=\widetilde\eta(H)\v$, one verifies that
$
\eta\in C^\infty_{\rm c}\big((a,b);[0,1]\big)\subset
C^\infty_{\rm c}\big(\R\setminus \kappa(H)\big)
$ and
that $\psi=\eta(H)\psi$, which imply that $\psi\in\D_0$. Moreover, one has
\begin{align*}
\|\v-\psi\|&\le\big\|[1-\widetilde{\eta}(H)]E^H([a,b])\v\big\|
+\big\|[1-\widetilde{\eta}(H)]\big[1-E^H([a,b])\big]\v\big\|\\
&\le\big\|[1-\widetilde{\eta}(H)]E^H(L_\delta)\v\big\|
+\big\|\big[1-E^H([a,b])\big]\v\big\|\\
&\textstyle\le\frac\varepsilon2+\frac\varepsilon2\,.
\end{align*}
Thus $\|\v-\psi\|\le\varepsilon$ for $\psi\in\D_0$, and the claim is proved.

(b) If $\sigma_{\rm p}(H)=\varnothing$, then it follows from the above discussion
that $\H_{\rm ac}(H)=\H$. In view of what precedes, it is enough to show that the
vector $\psi\equiv\widetilde{\eta}(H)\v$ of point (a) belongs to $\dom(\<\Phi\>^t)$:
The operator $\widetilde{\eta}(H)$ belongs to $C^k(\Phi)$, since $H$ is of class
$C^k(\Phi)$ and $\widetilde{\eta}\in C^\infty_{\rm c}(\R)$ (see
\cite[Thm.~6.2.5]{ABG}). So, we obtain from \cite[Prop.~5.3.1]{ABG} that $\<\Phi\>^t\widetilde{\eta}(H)\<\Phi\>^{-t}$ is bounded on $\H$, which implies the
claim.
\end{proof}

\section{Examples}\label{SecEx}
\setcounter{equation}{0}

In this section we show that Assumptions \ref{chirimoya} and \ref{commute} are
satisfied in various general situations. In these situations all the results of the preceding sections such as Theorem \ref{not_bad} or Formula \eqref{bisoufe} hold.
However, it is usually impossible to determine explicitly the set $\kappa(H)$ when
the framework is too general. Therefore, we also illustrate our approach with some
concrete examples for which everything can be computed explicitly. When possible, we
also relate these examples with the different cases presented in Section \ref{Interpretation}. For that purpose, we shall always assume that $f$ is a real and
even function in $\S(\R^d)$ with $f=1$ on a neighbourhood of $0$.

The configuration space of the system under consideration will sometimes be $\R^n$,
and the corresponding Hilbert space $\ltwo(\R^n)$. In that case, the notations $Q\equiv(Q_1,\ldots,Q_n)$ and $P\equiv(P_1,\ldots,P_n)$ refer to the families of
position operators and momentum operators. More precisely, for suitable
$\v\in\ltwo(\R^n)$ and each $j\in\{1,\ldots,n\}$, we have $(Q_j\v)(x)=x_j\v(x)$ and $(P_j\v)(x)=-i(\partial_j\v)(x)$ for each $x\in\R^n$.

\subsection{$\boldsymbol{H'}$ constant}

Suppose that $H$ is of class $C^1(\Phi)$, and assume that there exists
$v\in\R^d\setminus\{0\}$ such that $H'=v$. Then $H$ is of class $C^\infty(\Phi)$,
Assumption \ref{chirimoya} is directly verified, and one has on~$\dom(H)$
$$
H(x)=H(0)+\int_0^1\d t\,\big(x\cdot H'(tx)\big)
=H+\int_0^1\d t\,\e^{-itx\cdot\Phi}\big(x\cdot H'\big)\e^{-itx\cdot\Phi}
=H+x\cdot v.
$$
This implies Assumption \ref{commute}. Furthemore $\kappa(H)=\varnothing$, and
$\sigma(H)=\sigma_{\rm ac}(H)$ due to Theorem \ref{not_bad}. So, the set $\D_t$ is
dense in $\H$ for each $t\geq 0$, due to Lemma \ref{density}.(b). The operator
$R_f'(H')$ reduces to the constant vector $R_f'(v)$. Therefore, we have the equality
$T_f=-R_f'(v)\cdot\Phi$ on $\D_1$, and it is easily shown that $T_f$ is essentially
self-adjoint on $\D_1$. It follows from the case $1$ of Section \ref{Interpretation}
that the spectrum of $H$ covers the whole real line, and there exists a unitary
operator $\U:\H\to\ltwo(\R;\C^N,\d\lambda)$ such that
$$
\langle\psi,T_f\v\rangle
=\int_\R\d\lambda\,\textstyle\big\langle(\U\psi)(\lambda),
i\,\frac{\d(\U\v)}{\d\lambda}(\lambda)\big\rangle_{\C^N}
$$
for each $\psi\in\H$ and $\v\in\D_1$.

Typical examples of operators $H$ and $\Phi$ fitting into this construction are
Friedrichs-type Hamiltonians and position operators. For illustration, we mention
the case $H:=v\cdot P+V(Q)$ and $\Phi:=Q$ in $\ltwo(\R^d)$, with
$v\in\R^d\setminus\{0\}$ and $V\in\linf(\R^d;\R)$ (see also \cite[Sec.~5]{Tie09_3}
for informations on quantum time delay in a similar case).

Stark Hamiltonians and momentum operators also fit into the construction, \ie $H:=P^2+v\cdot Q$ in $\ltwo(\R^d)$ with $v\in\R^d\setminus\{0\}$, and $\Phi:=P$. We refer to \cite{Raz71,RW89,RW91} for previous accounts
on the theory of time operators and quantum time delay in similar situations.

Note that these first two examples are interesting since the operators $H$ contain
not only a kinetic part, but also a potential perturbation.

Another example is provided by the Jacobi operator related to the family of Hermite
polynomials (see \cite[Appendix A]{Sah08} for details). In the Hilbert space
$\H:=\ell^2(\N)$, consider the Jacobi operator given for $\v\in\H$ by
$$
\textstyle
(H\v)(n):=\frac{\sqrt{n-1}}2\,\v(n-1)+\frac{\sqrt{n}}2\,\v(n+1)
$$
with the convention that $\v(0)=0$. The operator $H$ is essentially self-adjoint on
$\ell^2_0$, the subspace of sequences in $\H$ with only finitely many non-zero
components. As operator $\Phi$ (with one component), take
$$
(\Phi\v)(n):=-i\big\{\sqrt{n-1}\,\v(n-1)-\sqrt n\,\v(n+1)\big\},
$$
which is also essentially self-adjoint on $\ell^2_0$. Then $H$ is of class
$C^1(\Phi)$ and $H'\equiv i[H,\Phi]=1$, and so the preceding results hold.

\subsection{$\boldsymbol{H'=H}$}\label{H'=H}

Suppose that $\Phi$ has only one component, and assume that $H$ is $\Phi$-homogeneous
of degree $1$, \ie $H(x)\equiv\e^{-ix\Phi}H\e^{ix\Phi}=\e^xH$ for all $x\in\R$.
This implies that $H$ is of class $C^\infty(\Phi)$ and that $H'=H$. So, Assumptions \ref{chirimoya} and \ref{commute} are readily verified. Moreover, since
$\kappa(H)=\{0\}$, Theorem \ref{not_bad} implies that $H$ is purely absolutely
continuous except at the origin, where it may have the eigenvalue $0$.

Now, let us show that the formal formula of Remark \ref{minusculeremarque} holds in
this case. For any $\v\in\D_1$ one has by Remark \ref{petiteremarque} that
$R'_f(H')\v\equiv R_f'(H)\v$ belongs to $\dom(\Phi)$. On another hand, we have
$$
\Phi\v=\big\{H\Phi+[\Phi,H]\big\}H^{-1}\v=H(\Phi+i)H^{-1}\v,
$$
which implies that
$R_f'(H)\Phi\v
=R_f'\big(\frac H{|H|}\big)\frac H{|H|}(\Phi+i)H^{-1}\v\in\H$.
In consequence, the operator
$$
T_f=-\12 \big(\Phi R_f'(H)+R_f'(H)\Phi\big)
$$
is well-defined on $\D_1$. In particular, if $0$ is not an eigenvalue of $H$, then
$T_f$ is a symmetric operator and the discussion of the case $2$ of Section
\ref{Interpretation} is relevant (if $T_f$ is essentially self-adjoint, the case $1$
is relevant).

We now give two examples of pairs $\{H,\Phi\}$ satisfying the preceding assumptions.
Other examples are presented in \cite[Sec.~10]{BG91}. Suppose that $H:=P^2$ is the
free Schr\"odinger operator in $\H:=\ltwo(\R^n)$ and $\Phi:=\frac14(Q\cdot P+P\cdot Q)$
is the generator of dilations in $\H$. Then the relation
$\e^{-ix\Phi}H\e^{ix\Phi}=\e^xH$ is satisfied, $\sigma(H)=\sigma_{\rm ac}(H)=[0,\infty)$.
Furthermore, for $\psi\in\H$ and
$\v\in\F C^\infty_{\rm c}\big(\R^n\setminus \{0\}\big)\subset\D_1$ a direct calculation
using Formula \eqref{minusone} shows that
$$
\langle\psi,T_f\v\rangle
=\big\langle\psi,{\textstyle\frac14}\big(Q\cdot PP^{-2}+PP^{-2}\cdot Q\big)\v\big\rangle
=\int_0^\infty\d\lambda\,\textstyle\big\langle(\U\psi)(\lambda),
i\,\frac{\d(\U\v)}{\d\lambda}(\lambda)\big\rangle_{\ltwo(\mathbb S^{n-1})},
$$
where $\U:\H\to\int_{[0,\infty)}^\oplus\d\lambda\,\ltwo(\mathbb S^{n-1})$ is the
spectral transformation for $P^2$. This example corresponds to the case $2$ of
Section \ref{Interpretation}.

Another example of $\Phi$-homogeneous operator is provided by the Jacobi operator
related to the family of Laguerre polynomials (see \cite[Appendix A]{Sah08} for
details). In the Hilbert space $\H:=\ell^2(\N)$, consider the Jacobi operator
given for $\v\in\H$ by
$$
(H\v)(n):=(n-1)\v(n-1)+(2n-1)\v(n)+n\v(n+1),
$$
with the convention that $\v(0)=0$. The operator $H$ is essentially self-adjoint on $\ell^2_0$. As operator $\Phi$ (with one component), take
$$
\textstyle
(\Phi\v)(n):=-\frac i2\big\{(n-1)\v(n-1)-n\v(n+1)\big\}.
$$
Then one has $H'\equiv i[H,\Phi]=H$, which implies that $H$ is $\Phi$-homogeneous
of degree $1$ and so the preceding results hold.

\subsection{Dirac operator}

In the Hilbert space $\H:=\ltwo(\R^3;\C^4)$ we consider the Dirac operator for a
spin-$\12$ particle of mass $m>0$
$$
H:=\alpha\cdot P+\beta m,
$$
where $\alpha\equiv(\alpha_1,\alpha_2,\alpha_3)$
and $\beta$ denote the usual $4\times 4$ Dirac matrices.
It is known that $H$ has domain $\H^1(\R^3;\C^4)$, that $|H|=(P^2+m^2)^{1/2}$ and that
$\sigma(H)=\sigma_{\rm ac}(H)=(-\infty,-m]\cup [m,\infty)$.

We also let
$
\Phi:=\U_{\rm FW}^{-1}Q\U_{\rm FW}\equiv Q_{\rm NW}
$
be the Wigner-Newton position operator, with $\U_{\rm FW}$ the usual
Foldy-Wouthuysen transformation \cite[Sec.~1.4.3]{Tha92}. Then
a direct calculation shows that
$$
\textstyle
H(x)=\sqrt{\frac{(P+x)^2+m^2}{P^2+m^2}}\;\!H
$$
for each $x \in \R^3$, and thus Assumptions \ref{chirimoya} and
\ref{commute} are clearly
satisfied. Furthermore, since
$H_j'=P_jH^{-1}$
for each $j=1,2,3$, it follows that
$$
\textstyle
(H')^2=P^2H^{-2}=(H^2-m^2)H^{-2}.
$$
Clearly, $\ker\big((H')^2\big)=\{0\}$ and one infers from Definition \ref{surkappa}
that $\kappa(H)=\{\pm m\}$, and from Lemma \ref{density}.(b) that the sets
$$
\D_t=\big\{\v\in\U_{\rm FW}^{-1}\dom\big(\langle Q\rangle^t\big)\mid\eta(H)\v=\v
\textrm{ for some }\eta\in C^\infty_{\rm c}\big(\R\setminus\{\pm m\}\big)\big\},
$$
are dense in $\H$. So the discussion of the case $2$ of Section
\ref{Interpretation} is relevant.

We now show that the formal formula of Remark \ref{minusculeremarque} holds if $f$
is radial. Indeed, each $\v\in\D_1$ satisfies $\v=\eta(H)\U_{\rm FW}^{-1}\psi$ for
some $\eta\in C^\infty_{\rm c}\big(\R\setminus\{\pm m\}\big)$ and some
$\psi\in\dom(\langle Q\rangle)$. So, we have
$$
H'(H')^{-2}\cdot Q_{\rm NW}\v
=PP^{-2}H\cdot\U_{\rm FW}^{-1}Q\U_{\rm FW}\eta(H)\U_{\rm FW}^{-1}\psi
=\U_{\rm FW}^{-1}PP^{-2}\beta|H|\cdot Q\eta(\beta|H|)\psi\in\H,
$$
and the operator $T$ of \eqref{T} is symmetric and can be written on $\D_1$ in the
simpler form
$$
T=\12 \big\{Q_{\rm NW}\cdot H'(H')^{-2} + H'(H')^{-2}\cdot Q_{\rm NW}\big\}
\equiv \12 \big\{Q_{\rm NW}\cdot PP^{-2}H + PP^{-2}H\cdot Q_{\rm NW}\big\}.
$$

Now let $h:\R^3\to\R$ be defined by $h(\xi):=(\xi^2+m^2)^{1/2}$. Then it
is known that
$\U_{\rm FW}H\U_{\rm FW}^{-1}=\beta h(P)$, and a direct
calculation shows that
$$
\U_{\rm FW}T\U_{\rm FW}^{-1}
=\12 \beta\big\{Q\cdot PP^{-2}(P^2+m^2)^{1/2}+PP^{-2}(P^2+m^2)^{1/2}\cdot Q\big\}
\textstyle=\12\beta\big\{Q\cdot\frac{h'(P)}{h'(P)^2}
+\frac{h'(P)}{h'(P)^2}\cdot Q\big\}
$$
on $\U_{\rm FW}\D_1$. Furthermore there exists a spectral transformation
$\U_0:\ltwo(\R^3)\to\int_{[m,\infty)}^\oplus\d\lambda\,\ltwo(\mathbb S^2)$ for $h(P)$
such that
$$
\textstyle
\U_0\big\{Q\cdot\frac{h'(P)}{h'(P)^2}+\frac{h'(P)}{h'(P)^2}\cdot
Q\big\}\U_0^{-1}
$$
is equal to the operator $2i\frac\d{\d\lambda}$ of differentiation with respect to
the spectral parameter $\lambda$ of $h(P)$ (see \cite[Lemma 3.6]{Tie09_3} for a
precise statement). Combining the preceding transformations we obtain for each
$\psi\in\H$ and $\v\in\D_1$ that
$$
\langle\psi,T\v\rangle=\int_{\sigma(H)}\d\lambda\,\textstyle
\big\langle(\U\psi)(\lambda),i\;\!\frac{\d(\U\v)}{\d\lambda}(\lambda)
\big\rangle_{\ltwo(\mathbb S^2;\C^2)},
$$
where $\U:\H\to\int_{\sigma(H)}^\oplus\d\lambda\,\ltwo(\mathbb S^2;\C^2)$ is the
spectral transformation for the free Dirac operator $H$.

\subsection{Convolution operators on locally compact groups}
\label{colcg}

This example is partially inspired from \cite{MT07}, where the spectral nature of
convolution operators on locally compact groups is studied.

Let $G$ be a locally compact group with identity $e$ and a left Haar measure $\rho$.
In the Hilbert space $\H:=\ltwo(G,\d\rho)$ we consider the operator $H_\mu$ of
convolution by $\mu\in\M(G)$, where $\M(G)$ is the set of complex bounded Radon
measures on $G$. Namely, for $\v\in\H$ one sets
$$
(H_\mu\v)(g):=(\mu\ast\v)(g)\equiv\int_G\d\mu(h)\,\v(h^{-1}g)
\quad\hbox{for \aew $g\in G$},
$$
where the notation \aew stands for ``almost everywhere'' and refers to the Haar
measure $\rho$. The operator $H_\mu$ is bounded with norm $\|H_\mu\|\le|\mu|(G)$,
and it is self-adjoint if $\mu$ is symmetric, \ie $\mu(E)=\overline{\mu(E^{-1})}$ for
each Borel subset $E$ of $G$. For simplicity, we also assume that $\mu$ is central and
with compact support, where central means that $\mu(h^{-1}Eh)=\mu(E)$ for each
$h\in G$ and each Borel subset $E$ of $G$.

We recall that given two measures $\mu,\nu\in\M(G)$, their convolution
$\mu\ast\nu\in\M(G)$ is defined by the relation \cite[Eq. 2.34]{Fo95}
$$
\int_G\d(\mu\ast\nu)(g)\,\psi(g):=\int_G\int_G\d\mu(g)\d\nu(h)\,\psi(gh)
\qquad\forall\psi\in C_0(G),
$$
where $C_0(G)$ denotes the $C^*$-algebra of continuous complex functions on $G$
vanishing at infinity. If $\mu\in\M(G)$ has compact support and $\zeta:G\to\C$ is
continuous, then the linear functional
$$
C_0(G)\ni\psi\mapsto\int_G\d\mu(g)\,\zeta(g)\psi(g)\in\C
$$
is bounded, and there exists a unique measure with compact support associated with
it, due to the Riesz-Markov representation theorem. We write $\zeta\mu$ for this
measure.

A natural choice for the family of operators $\Phi\equiv(\Phi_1,\ldots,\Phi_d)$
are, if they exist, real characters $\Phi_j\in\Hom(G;\R)$, \ie continuous group
morphisms from $G$ to $\R$. With this choice, one obtains that
$$
[H_\mu(x)\v](g)
\equiv\big(\e^{-ix\cdot\Phi}H_\mu\e^{ix\cdot\Phi}\v\big)(g)
=\int_G\d\mu(h)\,\e^{-ix\cdot\Phi(h)}\v(h^{-1}g)
$$
for each $x\in\R^d$, $\v\in\H$, and \aew $g\in G$. Namely, $H_\mu(x)$ is equal to the
operator of convolution by the measure $\e^{-ix\cdot\Phi}\mu$, \ie
$H_\mu(x)=H_{\e^{-ix\cdot\Phi}\mu}$. Since $\mu$ has compact support and each
$\Phi_j$ is continuous, this implies that $H_\mu$ is of class $C^\infty(\Phi)$. So
Assumption \ref{chirimoya} is satisfied. Furthermore, the commutativity of central
measures with respect to the convolution product implies that $\mu\ast\e^{-ix\cdot\Phi}\mu=\e^{-ix\cdot\Phi}\mu\ast\mu$ or equivalently that
$HH(x)=H(x)H$. So Assumption \ref{commute} is satisfied. Finally, the equality
$H_\mu(x)=H_{\e^{-ix\cdot\Phi}\mu}$ readily implies that $(H_\mu')_j=H_{-i\Phi_j\mu}$.

Since both Assumptions \ref{chirimoya} and \ref{commute} are satisfied, the general
results of the previous sections apply. However, it is very complicated to describe
the set $\kappa(H_\mu)$ in the present generality. Therefore, we shall now assume
that the group $G$ is abelian in order to use the Fourier transformation to determine
some properties of $\kappa(H_\mu)$. So let us assume that $G$ is a locally compact
abelian group. Then any measure on $G$ is automatically central, and thus we only
need to suppose that $\mu$ is symmetric and with compact support. For a suitably
normalised Haar measure $\rho_\land$ on the dual group $\widehat G$, the Fourier
transformation $\F$ defines a unitary isomorphism from $\H$ onto
$\ltwo(\widehat G,\d\rho_\land)$. It maps unitarily $H_\mu$ on the operator $M_m$ of
multiplication with the bounded continuous real function $m:=\F(\mu)$ on $\widehat G$.
Furthermore, one has
\begin{equation}\label{spec_ponc}
\sigma(H_\mu)=\sigma(M_m)=\overline{m(\widehat G)},\quad
\sigma_{\rm p}(H_\mu)=\sigma_{\rm p}(M_m)
=\overline{\left\{s\in\R\mid\rho_\land\(m^{-1}(s)\)>0\right\}},
\end{equation}
where the overlines denote the closure in $\R$.

Let us recall that there is an almost canonical identification of $\Hom(G,\R)$ with
the vector space $\Hom(\R,\widehat G)$ of all continuous one-parameter subgroups of
$\widehat G$. Given the real character $\Phi_j$, we denote by
$\Upsilon_j\in\Hom(\R,\widehat G)$ the unique element satisfying
$$
\big\langle g,\Upsilon_j(t)\big\rangle=\e^{it\Phi_j(g)}
\quad\hbox{for all }t\in\R\hbox{ and }g\in G,
$$
where $\<\cdot,\cdot\>:G\times\widehat G\to\C$ is the duality between $G$ and
$\widehat G$.

\begin{Definition}\label{riss}
{\rm A function $m:\widehat G\to\C$ is \emph{differentiable at $\xi\in\widehat G$
along the one-parameter subgroup $\Upsilon_j\in\Hom(\R,\widehat G)$} if the
function $\R\ni t\mapsto m\big(\xi+\Upsilon_j(t)\big)\in\C$ is differentiable
at $t=0$. In such a case we write $(d_jm)(\xi)$ for
$\frac\d{\d t}\,m\big(\xi+\Upsilon_j(t)\big)\big\vert_{t=0}$. Higher order
derivatives, when existing, are denoted by $d_j^km$, $k\in\N$.}
\end{Definition}

We refer to \cite{Ri53} for more details on differential calculus on locally
compact groups. Here we only note that (since $\mu$ has compact support) the
function $m=\F(\mu)$ is differentiable at any point $\xi$ along the one-parameter
subgroup $\Upsilon_j$, and $-i\F(\Phi_j\mu)=d_jm$ \cite[p.~68]{Ri53}. This
implies that the operator $(H_\mu')_j$ is mapped unitarily by $\F$ on the
multiplication operator $M_{d_jm}$, and thus $(H_\mu')^2$ is unitarily equivalent
to the operator of multiplication by the function $\sum_j(d_jm)^2$. It follows
that
$$
\textstyle
\kappa(H_\mu)\supset\big\{\lambda\in\R\mid\exists\xi\in\widehat G
\hbox{ such that }m(\xi)=\lambda\hbox{ and }\sum_j(d_jm)(\xi)^2=0\big\}.
$$

This property of $\kappa(H_\mu)$ suggests a way to justify the formal formula
of Remark \ref{minusculeremarque} and to write nice formulas for the operator
$T$ given by \eqref{T}. Indeed, since $\F\Phi_j\F^{-1}$ acts as the differential
operator $id_j$ in $\ltwo(\widehat G,\d\rho_\land)$, it follows that $\Phi_j$
leaves invariant the complement of the support of the functions on which it acts.
Therefore, the set $\Phi_j\D_1\equiv\F^{-1}(id_j)\F\D_1$ is included in the domain
of the operator
$$
\textstyle
\frac{(H_\mu')_j}{(H_\mu')^2}\equiv\F^{-1}\frac{M_{d_jm}}{M_{\sum_k(d_km)^2}}\,\F.
$$
Thus the formula \eqref{T} takes the form
$$
\textstyle
T=\12 \sum_j\Big\{\Phi_j\frac{H_{-i\Phi_j\mu}}{\sum_k(H_{-i\Phi_k\mu})^2}
+\frac{H_{-i\Phi_j\mu}}{\sum_k(H_{-i\Phi_k\mu})^2}\,\Phi_j\Big\}
$$
on $\D_1$, or alternatively the form
\begin{equation}\label{T_f_conv}
\textstyle
\F T\F^{-1}=\frac{i}{2}\sum_j\Big\{d_j\frac{M_{d_jm}}{M_{\sum_k(d_km)^2}}
+\frac{M_{d_jm}}{M_{\sum_k(d_km)^2}}\,d_j\Big\}
\end{equation}
on $\F\D_1$ (note that the last expression is well-defined on $\F\D_1$, since
$m=\F(\mu)$ is of class $C^2$ in the sense of Definition \ref{riss}).

In simple situations, everything can be calculated explicitly. For instance, when
$G=\Z^d$, the Haar measure $\rho$ is the counting measure, and the most
natural real characters $\Phi_j$ are the position operators given by
$$
(\Phi_j\v)(g):=g_j\v(g),\qquad\v\in\ltwo(\Z^d),
$$
where $g_j$ is the $j$-th component of $g\in\Z^d$. The operators $H_\mu$ and
$(H_\mu')^2$ are unitarily equivalent to multiplication operators on
$\widehat G=(-\pi,\pi]^d$. Since the measures $\mu$ and $\Phi_j\mu$ have compact
(and thus finite) support, these operators are just multiplication operators by
polynomials of finite degree in the variables $\e^{-i\xi_1},\ldots,\e^{-i\xi_d}$,
with $\xi_j\in(-\pi,\pi]$. So, the set $\kappa(H_\mu)$ is finite, and the
characterisation \eqref{spec_ponc} of the point spectrum of $H_\mu$ implies that $\sigma_{\rm p}(H_\mu)=\varnothing$ if $\supp(\mu)\neq\{e\}$. By taking into
account Lemma \ref{density}.(b) and Theorem \ref{not_bad}, we infer that the sets
$\D_t$ are dense in $\H$ for each $t\geq0$, and thus the case $2$ of Section \ref{Interpretation} applies. Finally, we mention as a corollary the following
spectral result:

\begin{Corollary}
Let $\mu$ be a symmetric measure on $\Z^d$ with finite support. If
$\supp(\mu)\neq\{e\}$, then the convolution operator $H_\mu$ in $\H:=\ltwo(\Z^d)$ is
purely absolutely continuous.
\end{Corollary}

\subsection{$\boldsymbol{H=h(P)}$}

Consider in $\H:=\ltwo(\R^d)$ the dispersive operator $H:=h(P)$, where
$h\in C^3(\R^d;\R)$ satisfies the following condition: For each multi-indices
$\alpha,\beta\in\N^d$ with $\alpha>\beta$, $|\alpha|=|\beta|+1$, and
$|\alpha|\leq 3$, we have
\begin{equation}\label{Cond_h}
|\partial^\alpha h|\le{\rm Const.}\;\!\big(1+|\partial^\beta h|\big).
\end{equation}
Note that this class of operators $h(P)$ contains all the usual elliptic free
Hamiltonians appearing in physics.

Take for the family $\Phi\equiv(\Phi_1,\ldots,\Phi_d)$ the position operators
$Q\equiv(Q_1,\ldots,Q_d)$. Then we have for each $x\in\R^d$
$$
H(x)=\e^{-ix\cdot Q}H_\mu\e^{ix\cdot Q}=h(P+x),
$$
and $H'=h'(P)$. So Assumption \ref{commute} is directly verified and Assumption \ref{chirimoya} follows from \eqref{Cond_h}. Therefore all the results of the
previous sections are valid. We do not give more details since many aspects of this
example, including the existence of time delay, have already been extensively
discussed in \cite{Tie09_3}. We only add some comments in relation with the case $3$
of Section \ref{Interpretation}.

Assume that there exist $\lambda\in\R$ and a maximal subset $\Omega\subset\R^d$ of
strictly positive Lebesgue measure such that $h(x)=\lambda$ for all $x\in\Omega$.
Then any $\v$ in $\H_\Omega:=\{\psi\in\H\mid\supp(\F\psi)\subset\Omega\}$ is an
eigenvector of $h(P)$ with eigenvalue $\lambda$. Furthermore, one has
$\F^{-1}\H_\Omega\subset\K\equiv\ker\big(h'(P)^2\big)$, and for simplicity we assume
that the first inclusion is an equality. Then, an application of the Fourier
transformation shows that $Q_j\D_1\subset\G$ for each $j$, where $\G$ is the
orthocomplement of $\K$ in $\H$. Thus Remark \ref{queneni} applies, and one can
consider the restrictions of $H$ and $T_f$ to the subspace $\G$, as described in the
case $3$ of Section \ref{Interpretation}. In favorable situations, we expect that the restriction of $T_f$ to $\G$ acts as $i\frac\d{\d\lambda}$ in the spectral
representation of the restriction of $H$ to $\G$.

\subsection{Adjacency operators on admissible graphs}

Let $(X,\sim)$ be a graph $X$ with no multiple edges or loops. We write $g\sim h$
whenever the vertices $g$ and $h$ of $X$ are connected. In the Hilbert space
$\H:=\ell^2(X)$ we consider the adjacency operator
\begin{equation*}
(H\v)(g):=\sum_{h\sim g}\v(h),\quad\v\in\H,~g\in X.
\end{equation*}
We denote by $\deg(g):=\#\{h\in X\mid h\sim g\}$ the degree of the vertex $g$. Under
the assumption that $\deg(X):=\sup_{g\in X}\deg(g)$ is finite, $H$ is a bounded
self-adjoint operator in $\H$. The spectral analysis of the adjacency operator on
some general graphs has been performed in \cite{MRT07}. Here we consider only a
subclass of such graphs called admissible graphs.

A directed graph $(X,\sim,<)$ is a graph $(X,\sim)$ and a relation $<$ on the graph
such that, for any $g,h\in X$, $g\sim h$ is equivalent to $g<h$ or $h<g$, and one
cannot have both $h<g$ and $g<h$. We also write $h>g$ for $g<h$. For a fixed $g$, we
denote by $N^-(g)\equiv\{h\in X\mid g<h\}$ the set of fathers of $g$ and by
$N^+(g)\equiv\{h\in X\mid h<g\}$ the set of sons of $g$. The set
$\{h\in X\mid g\sim h\}$ of neighbours of $g$ is denoted by
$N(g)\equiv N^-(g)\cup N^+(g)$. When using drawings, one has to choose a direction
(an arrow) for any edge. By convention, we set $g\leftarrow h$ if $g<h$, \ie any
arrow goes from a son to a father. When directions have been fixed, we use the
simpler notation $(X,<)$ for the directed graph $(X,\sim,<)$.

\begin{Definition}\label{admisibil}
{\rm A directed graph $(X,<)$ is called admissible if
\begin{enumerate}
\item[(a)] any closed path in $X$ has index zero (the index of a path is the
difference between the number of positively oriented edges in the path and that of
the negatively oriented ones),
\item[(b)] for any $g,h\in X$, one has
$\#\{N^-(g)\cap N^-(h)\}=\#\{N^+(g)\cap N^+(h)\}$.
\end{enumerate}}
\end{Definition}

It is proved in \cite[Lemma~5.3]{MRT07} that for admissible graphs there exists a
unique (up to constant) map $\Phi:X\to\Z$ satisfying $\Phi(h)+1=\Phi(g)$ whenever
$h<g$. With this choice of operator $\Phi$, one obtains that
\begin{equation}\label{bonbon}
[H(x)\v](g)=\sum_{h\sim g}\e^{ix[\Phi(h)-\Phi(g)]}\v(h)
\end{equation}
for each $x\in\R$, $\v\in\H$, and $g\in X$. Therefore, the commutativity of $H$ and
$H(x)$ is equivalent to the condition
\begin{equation*}
\sum_{h\in N(g)\cap N(\ell)}\big(\e^{ix[\Phi(\ell)-\Phi(h)]}
-\e^{ix[\Phi(h)-\Phi(g)]}\big)=0
\end{equation*}
for each $g,\ell\in X$. By taking into account the growth property of $\Phi$ and
Hypothesis (b) of Definition \ref{admisibil}, one obtains that the parts
$h\in N^-(g)\cap N^-(\ell)$ and $h\in N^+(g)\cap N^+(\ell)$ of the sum are of
opposite sign, and that the parts $h\in N^-(g)\cap N^+(\ell)$ and
$h\in N^+(g)\cap N^-(\ell)$ are null. So Assumption \ref{commute} is satisfied.
One also verifies by using Formula \eqref{bonbon} that $H$ belongs to
$C^\infty(\Phi)$, and that Assumption \ref{chirimoya} holds. It follows that
the general results presented before apply.

Now, the operator $H'$ acts as
$
(H'\v)(g)=i\big(\sum_{h>g}\v(h)-\sum_{h<g}\v(h)\big),
$
and it is proved in \cite[Sec.~5]{MRT07} that
\begin{equation}\label{eigens}
\textstyle\H_{\rm p}(H)=\ker(H)=\ker(H')
=\big\{\v\in\H\mid\sum_{h>g}\v(h)=0=\sum_{h<g}\v(h)\hbox{ for each }g\in X\big\}.
\end{equation}
It is also proved that $H$ is purely absolutely continuous, except at the origin
where it may have an eigenvalue with eigenspace given by \eqref{eigens}. The
proof of these statements is based on the method of the weakly conjugate operator \cite{BKM96}.

However, in the present generality, it is hardly possible to obtain a simple
description of the set $\kappa(H)$ or the operator $T_f$. We refer then to \cite[Sec.~6]{MRT07} for explicit examples of admissible graphs with adjacency
operators whose kernels are either trivial or non trivial, and develop one example
for which more explicit computations can be performed. This example furnishes an
illustration of the discussion in the case $3$ of Section \ref{Interpretation}.

\begin{figure}[htbp]
\begin{center}
\input{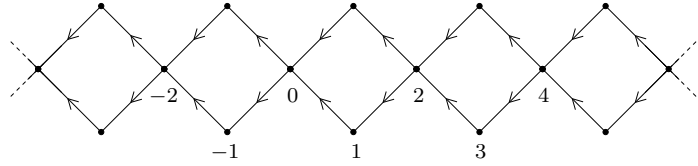}
\caption{Example of an admissible directed graph $X$}
\label{graph1}
\vspace{-10pt}
\end{center}
\end{figure}

We consider the admissible graph of Figure \ref{graph1}, and endow it with the
function $\Phi:X\to\Z$ as shown on the picture. The vertices of the graph are
denoted by $z_-$ and $z_+$ when $\Phi$ takes an odd value, and by $z$ when
$\Phi$ takes an even value. More precisely, $\Phi(z)=z$ for $z$ even, and $\Phi(z_-)=\Phi(z_+)=z$ for $z$ odd. By using \eqref{eigens}, it is easily
observed that $\K\equiv\ker\big((H')^2\big)$ is equal to
$$
\big\{\v\in\ltwo(X)\mid\v(z)=0\hbox{ for }z\hbox{ even, and }\v(z_-)=-\v(z_+)
\hbox{ for }z\hbox{ odd}\big\}.
$$
On the other hand, the orthocomplement $\G$ of $\K$ in $\ltwo(X)$ is unitarily
equivalent to $\ell^2(\Z)$, and the restriction $\HH$ of $H$ to $\G$ is
unitarily equivalent to the operator in $\ell^2(\Z)$ defined by
$$
\big(\widetilde\HH\v\big)(z)
:=\sqrt2\big\{\v(z-1)+\v(z+1)\big\},\qquad\v\in\ell^2(\Z).
$$
Using the Fourier transformation, one shows that this operator is unitarily
equivalent to the multiplication operator $M$ in $\ltwo\big((-\pi,\pi]\big)$
given by the function $(-\pi,\pi]\ni\xi\mapsto2\sqrt2\cos(\xi)$.

Now, the operator $\Phi$ in $\ltwo(X)$ is clearly reduced by the decomposition
$\K\oplus\G$. As mentioned in Remark \ref{queneni}, this implies that the
operator $T_f$ is also reduced by this decomposition. By taking Formula
\eqref{T_f_conv} into account, one obtains that the restriction $\TT_{\!f}$ of
$T_f$ to $\G$ is unitarily equivalent to the operator
$$
\textstyle
\frac{i}{2}\Big\{\frac\d{\d\xi}\big[-2\sqrt{2}\sin(\xi)\big]^{-1}
+\big[-2\sqrt{2}\sin(\xi)\big]^{-1}\frac\d{\d\xi}\Big\}
$$
on $\F\D_1\subset\ltwo\big((-\pi,\pi]\big)$. This implies, as expected, that
$\TT_{\!f}$ acts as $i\frac\d{\d\lambda}$ in the spectral representation of $\HH$.

\subsection{Direct integral operators}

Let $\Omega$ be a measurable subset of $\R^n$ and let us consider a direct integral
$$
\H:=\int_\Omega^\oplus\d\xi\,\H_\xi,
$$
where $\d\xi$ is the usual Lebesgue measure on $\R^n$ and $\H_\xi$ are Hilbert
spaces. Take a decomposable self-adjoint operator
$H\equiv\int_\Omega^\oplus\d\xi\,H(\xi)$ in $\H$. Assume that there exists a
family $\Phi\equiv(\Phi_1,\ldots,\Phi_d)$ of operators in $\H$ such that Assumption
\ref{chirimoya} is satisfied. Assume also for each $x\in\R^d$ that the operator $H(x)$
defined by \eqref{H(x)} is decomposable, \ie there exists a family of self-adjoint
operators $H(\xi,x)$ in $\H_\xi$ such that $H(x)=\int_\Omega^\oplus\d\xi\,H(\xi,x)$.
Finally, assume that the operators $H(\xi)$ and $H(\xi,x)$ commute for each $x\in\R^d$
and \aew $\xi\in\Omega$, so that $H$ and $H(x)$ commute. Then Assumption \ref{commute}
holds, and the general theory developed in the preceding sections applies. Moreover,
it is easily observed that the fibered structure of the map $x\mapsto H(x)$ implies
that the operators $H'_j$ are also decomposable. Therefore, there exists for
each $j\in\{1,\ldots,d\}$  a family of self-adjoint operators $H'_j(\xi)$ such that
$H'_j=\int_\Omega^\oplus\d\xi\,H'_j(\xi)$. In consequence $\lambda\in\R$ is a regular
value of $H$ if there exists $\delta>0$ and $\textsc c<\infty$ such that
\begin{equation}\label{fibre}
\lim_{\varepsilon\searrow0}\big\|\big[\big(H'(\xi)\big)^2+\varepsilon\big]^{-1}
E^{H(\xi)}(\lambda;\delta)\big\|_{\H_\xi}<\textsc c
\end{equation}
for \aew $\xi\in\Omega$. We also recall that $\ker\big((H')^2\big)\neq\{0\}$ if and
only if there exists a measurable subset $\Omega_0\subset\Omega$ with positive
measure such that $\ker\big(H'(\xi)^2\big)\neq\{0\}$ for each $\xi\in\Omega_0$.

We now give an example of quantum waveguide-type fitting into this setting (see
\cite{Tie06} for more details). Let $\Sigma$ be a bounded open connected set in
$\R^m$, and consider in the Hilbert space $\ltwo(\Sigma\times\R)$ the Dirichlet
Laplacian $-\Delta_{\rm D}$. The partial Fourier transformation along the
longitudinal axis sends the initial Hilbert space onto the direct integral
$\H:=\int_\R^\oplus\d\xi\,\H_0$, with $\H_0:=\ltwo(\Sigma)$, and it sends
$-\Delta_{\rm D}$ onto the fibered operator $H:=\int_\R^\oplus\d\xi\,H(\xi)$,
with $H(\xi):=\xi^2-\Delta^\Sigma_{\rm D}$. Here, $-\Delta_{\rm D}^\Sigma$
denotes the Dirichlet Laplacian in $\Sigma$. By Choosing for $\Phi$ the position
operator $Q$ along the longitudinal axis one obtains that
$H(x)=\int_\R^\oplus\d\xi\,H(\xi,x)$ with
$H(\xi,x)=(\xi+x)^2-\Delta_{\rm D}^\Sigma$. Clearly, $H(\xi)$ and $H(\xi,x)$
commute, and so do $H$ and $H(x)$. Furthermore, the operator $H$ is of class
$C^\infty(\Phi)$, and $H'$ is the fibered operator given by $H'(\xi)=2\xi$. It
follows that both Assumptions \ref{chirimoya} and \ref{commute} hold, and thus
the general theory applies. Now a simple calculation using \eqref{fibre} shows
that $\kappa(H)=\sigma(-\Delta_{\rm D}^\Sigma)$. Furthermore, in the tensorial
representation $\ltwo(\Sigma)\otimes\ltwo(\R)$ of $\ltwo(\Sigma\times\R)$, one
obtains that $T_f=T={\textstyle \frac{1}{4}}\otimes(QP^{-1}+P^{-1}Q)$ on the
dense set
$$
\D_1=\big\{\v\in\ltwo(\Sigma)\otimes\dom(\langle Q\rangle)\mid
\v=\eta(-\Delta_{\rm D})\v\textrm{ for some }\eta\in
C^\infty_{\rm c}\big(\R\setminus \kappa(H)\big)\big\},
$$
and $T_f$ is equal to $i\frac\d{\d\lambda}$ in the spectral representation of
$-\Delta_{\rm D}$. In \cite{Tie06} it is even shown that the quantum time delay
exists and is given by Formula \eqref{Eisenbud} for appropriate scattering pairs
$\{-\Delta_{\rm D},-\Delta_{\rm D}+V\}$.

\section*{Acknowledgements}

S. Richard is supported by the Swiss National Science Foundation. R. Tiedra de
Aldecoa is partially supported by the N\'ucleo Cient\'ifico ICM P07-027-F
``Mathematical Theory of Quantum and Classical Magnetic Systems'' and by the Chilean
Science Foundation Fondecyt under the Grant 1090008.


\def\cprime{$'$}

\end{document}